\documentclass[10pt,journal,compsoc]{IEEEtran}
\usepackage{cite}
\usepackage{amsthm}
\usepackage{amsmath,amssymb,amsfonts}
\usepackage{algorithmic}
\usepackage{graphicx}
\usepackage{textcomp}
\usepackage{xcolor}

\usepackage{hyperref}
\usepackage{subfigure}
\usepackage{booktabs}
\usepackage{multirow}
\usepackage{multicol}
\usepackage[ruled,linesnumbered]{algorithm2e}

\newtheorem{theorem}{\textbf{Theorem}}

\newtheorem{lemma}{\textbf{Lemma}}
\newtheorem{definition}{\textbf{Definition}}

\def\BibTeX{{\rm B\kern-.05em{\sc i\kern-.025em b}\kern-.08em
    T\kern-.1667em\lower.7ex\hbox{E}\kern-.125emX}}

\ifCLASSINFOpdf
\else
\fi
\begin{document}
%
\title{Efficient Federated Learning with Enhanced Privacy via Lottery Ticket Pruning \\in Edge Computing}
%
%
%
%

\author{Yifan~Shi,
        Kang~Wei,~\IEEEmembership{ Member,~IEEE,}
        Li~Shen,
        Jun~Li, ~\IEEEmembership{Senior Member,~IEEE,}
        Xueqian~Wang, ~\IEEEmembership{Member,~IEEE,}
        Bo~Yuan, ~\IEEEmembership{Senior Member,~IEEE}, and Song Guo,~\IEEEmembership{Fellow,~IEEE}
\IEEEcompsocitemizethanks{\IEEEcompsocthanksitem Yifan Shi and Xueqian Wang are with the Center for Artificial Intelligence and Robotics, Shenzhen International Graduate School, Tsinghua University, 518055 Shenzhen, P.R. China (e-mail: shiyf21@mails.tsinghua.edu.cn; wang.xq@sz.tsinghua.edu.cn).

\IEEEcompsocthanksitem Li Shen is with JD Explore Academy, Beijing, China.
(e-mail: mathshenli@gmail.com)
\IEEEcompsocthanksitem Kang~Wei and Song Guo are with the Department of
Computing, Hong Kong Polytechnic University, Hong Kong 999077, China
(e-mail: \{song.guo, adam-kang.wei\}@polyu.edu.hk).
\IEEEcompsocthanksitem Jun~Li is with the School of Electrical and Optical Engineering, Nanjing University of Science and Technology, Nanjing, China. E-mail: jun.li@njust.edu.cn.

\IEEEcompsocthanksitem Bo Yuan is with Shenzhen Wisdom and Strategy Technology Co., Ltd., Shenzhen 518055, China (e-mail: boyuan@ieee.org).
}
}
%
%

\markboth{Journal of \LaTeX\ Class Files,~Vol.~14, No.~8, December~2022}%
{Shell \MakeLowercase{\textit{et al.}}: Bare Demo of IEEEtran.cls for Computer Society Journals}
%



\IEEEtitleabstractindextext{%

\begin{abstract}
Federated learning (FL) is a collaborative learning paradigm for decentralized private data from mobile terminals (MTs). However, it suffers from issues in terms of communication, resource of MTs, and privacy. Existing privacy-preserving FL methods usually adopt the instance-level differential privacy (DP), which provides a rigorous privacy guarantee but with several bottlenecks: severe performance degradation, transmission overhead, and resource constraints of edge devices such as MTs. To overcome these drawbacks, we propose Fed-LTP, an efficient and privacy-enhanced FL framework with \underline{\textbf{L}}ottery \underline{\textbf{T}}icket \underline{\textbf{H}}ypothesis (LTH) and zero-concentrated D\underline{\textbf{P}} (zCDP). It generates a pruned global model on the server side and conducts sparse-to-sparse training from scratch with zCDP on the client side. On the server side, two pruning schemes are proposed: (i) the weight-based pruning (LTH) determines the pruned global model structure; (ii) the iterative pruning further shrinks the size of the pruned model's parameters. Meanwhile, the performance of Fed-LTP is also boosted via model validation based on the Laplace mechanism. 
On the client side, we use sparse-to-sparse training to solve the resource-constraints issue and provide tighter privacy analysis to reduce the privacy budget. We evaluate the effectiveness of Fed-LTP on several real-world datasets in both independent and identically distributed (IID) and non-IID settings. The results clearly confirm the superiority of Fed-LTP over state-of-the-art (SOTA) methods in communication, computation, and memory efficiencies while realizing a better utility-privacy trade-off.

\end{abstract}

\begin{IEEEkeywords}
Federated Learning, Differential Privacy, Lottery Ticket Hypothesis, Zero-concentrated DP, Mobile Edge Computing
\end{IEEEkeywords}

}

\maketitle

\IEEEdisplaynontitleabstractindextext

%
\IEEEpeerreviewmaketitle

\IEEEraisesectionheading{\section{Introduction}\label{sec:introduction}}

%
%
%
%
Federated learning (FL) \cite{Li2020federated} allows distributed clients, e.g., mobile terminals (MTs), to collaboratively train a shared model under the orchestration of the cloud without sharing their local data\footnote{For instance, as a classic method in FL, Fed-Avg \cite{McMahan2017Communication} uses SGD to train the MTs selected in a distributed manner for multiple rounds in parallel, and then aggregates the model updates of each mobile to improve the global model's performance.}. 
However, FL faces several critical challenges, such as computational resources, memory, communication bandwidth, and privacy leakage \cite{Kairouz2021Advances}.
Most of recent works mainly focus on either the communication cost~\cite{McMahan2017Communication,sattler2019robust,hamer2020fedboost,dai2022dispfl} or resource overhead of MTs~\cite{xu2019elfish,wang2019adaptive,jiang2022model,huang2022achieving}. Furthermore, a curious server can also infer MTs’ privacy information such as membership and data features by well-designed generative models and/or shadow models~\cite{Fredrikson2015model,Shokri2017membership,Melis2018inference,Nasr2019comprehensive,zhang2022fine}. To address the privacy issue, differential privacy (DP) \cite{Dwork2014the}, the \emph{de-facto} standard in FL, can protect every instance in any mobile's dataset and the information between MTs (instance-level DP~\cite{Agarwal2018cpsgd,Hu2021federated, Sun2021federated,Sun2021practical}) or, less rigorously, only the information between MTs (client-level DP~\cite{McMahan2018learning,Robin2017differentially,KairouzL2021the,Rui2022Federated,Anda2022differentially, shi2023make}). 
For example, a bank needs an instance-level DP method to protect each data record of each customer from being identified, whereas a language prediction model in mobile devices only needs to protect the ownership of the data, and the client-level DP is sufficient. However, all DP methods introduce extra random noise proportional to the model size, which can lead to severe performance degradation, especially for instance-level DP. 

To mitigate the performance degradation and communication efficiency issues, existing instance-level DP techniques \cite{Agarwal2018cpsgd,Hu2021federated} use the local update sparsification method before uploading to improve the utility-privacy trade-off while reducing communication cost. Nevertheless, they still suffer from the following drawbacks: 1) only the communication cost of uploading (client-to-server) is reduced, without considering the server-to-client  cost; 2) the computational overhead and memory footprint of the mobile remains unchanged; 3) they only focus on the differentially private training without considering the model validation, and thus the resulting model is not necessarily optimal; 4) the sparsification method~\cite{Hu2021federated} has a large randomness and may cause performance degradation when the sparsity is high. Consequently, a critical question is: \emph{how to design privacy-preserving algorithm that can properly balance computation, memory efficiency of edge devices, and communication efficiency with improved model utility?}

To answer this question, we propose an efficient and privacy-enhanced \underline{\textbf{Fed}}erated learning framework with \underline{\textbf{L}}ottery \underline{\textbf{T}}icket \underline{\textbf{H}}ypo-
thesis (LTH) and zero-concentrated D\underline{\textbf{P}} (zCDP) method, named Fed-LTP. 
The key novelty of Fed-LTP lies in: \textbf{(i)} two server-side pruning schemes are designed to obtain a high-quality initial model: a weight-based pruning scheme to create a pretrained model also known as a winning ticket (WT), and a further iterative pruning scheme to create heterogeneous mobile models with different pruning degrees for further reducing the computation and
communication overheads; \textbf{(ii)} a server-side WT-broadcasting mechanism to ensure the stability and convergence of the global model while alleviating the large computational overhead and memory footprint of edge devices;
\textbf{(iii)} training the locally pruned model with zCDP to alleviate the privacy budget and using the Laplace mechanism based on the private validation dataset to get validation scores on the client side, which are then uploaded to the server for model validation to select the best global model and prevent over-fitting.

In summary, our main contributions are four-fold: 
\begin{itemize}
\item We are the first to introduce LTH into FL with DP and propose an efficient and privacy-enhanced FL framework (Fed-LTP), effectively alleviating the client-side resource constraints in terms of memory and computation while maintaining model utility and considering the two-way communication cost. 
\item We propose two server-side pruning schemes: a weight-based pruning scheme and a further iterative pruning scheme, to optimize the balance among utility, communication cost, and resource overhead of edge devices.  
\item We provide a new and tight privacy analysis (zCDP) on the privacy budget  for both training and validation data in each mobile to increase the level of privacy protection while maintaining the model utility/performance, thereby optimizing the utility-privacy trade-off.
\item  Compared with SOTA methods on various real-world datasets in both IID and non-IID settings, the effectiveness and superiority of our framework has been empirically validated.
\end{itemize}

Section II reviews the related work on instance-level DP and LTH in FL. Section III introduces the background of FL and DP. The proposed Fed-LTP is detailed in Section IV and the privacy analysis is conducted in Section V. Extensive experimental evaluation is presented in Section VI. This paper is concluded in Section VII with suggested directions for future work.

\section{Related Work}

\textbf{Instance-level DP in FL.} Recently, instance-level DP \cite{Hu2021federated,Sun2021federated,Sun2021practical} has been an emerging topic in FL. Fed-SPA~\cite{Hu2021federated} integrates random sparsification with gradient perturbation to obtain a better utility-privacy trade-off and reduce communication cost. Meanwhile, it uses the acceleration technique to ease the slow convergence issue. The federated model distillation framework FEDMD-NFDP \cite{Sun2021federated} can achieve improved performance under heterogeneous model architectures and eliminate the risk of white-box inference attacks by sharing model predictions. The work in \cite{Sun2021practical} studies the model aggregation of local differential privacy (LDP) and proposes an empirical solution to achieve a strict privacy guarantee for applying LDP to FL. 

\noindent
\textbf{Lottery Ticket Hypothesis.} LTH \cite{Frankle2019the} is a popular pruning method in a centralized machine learning setting. It generates the winning tickets (WTs) by iterative pruning, which allows for fast convergence close to the original model performance under the same training epochs. In the recent progress of LTH \cite{frankle2019stabilizing, frankle2020linear}, the two lines related to our work are the extension of LTH in FL and the extension of centralized machine learning with DP.
For instance, LotteryFL \cite{li2020lotteryfl} is a personalized and communication-efficient FL framework via exploiting LTH on non-IID datasets. HeteroFL \cite{diao2020heterofl} can be used to address heterogeneous MTs equipped with vastly different computation and communication capabilities. The work in \cite{itahara2020lottery} uses unlabeled public data to pretrain the model, and then uses LTH to compress the model for reducing the communication cost without affecting performance. CELL \cite{seo2021communication} extends LotteryFL by exploiting the downlink broadcast to improve communication efficiency. 
Compared with the above studies, our work is more closely related to PrunFL \cite{jiang2022model} with adaptive and distributed parameter pruning, which considers the limited resources of edge devices, and reduces both communication and computation overhead and minimizes the overall training time while maintaining a similar accuracy as the original model. By contrast, the combination of LTH and DP has been relatively less explored. DPLTM \cite{gondara2020differentially} uses “high-quality winners” and the custom score function for selection to improve the privacy-utility trade-off. Experimental studies show that DPLTM can achieve fast convergence, allowing for early stopping with reduced privacy budget consumption and reduced noise impact comparable to DPSGD \cite{Abadi2016Deep}.

Different from the existing works, our work is the first to introduce LTH into FL with public data to obtain an initial global model and network structure. In this way, the sparse structure can reduce the system costs (transmission and computation) and alleviate the performance degradation caused by the random noise injection. Meanwhile, the model validation with the Laplace mechanism is proposed to guarantee the performance of the final model.

\section{PRELIMINARY}
\begin{table*}[ht]
\centering
\small
\caption{Summary of main notation}
\label{notion}
\begin{tabular}{c|c} 
\toprule
\begin{tabular}[c]{@{}l@{}}$\mathcal D_{pub}, \mathcal D_{pr}$\\$\mathcal M$\\$\mathcal D, \mathcal D^{\prime}$\\$\epsilon, \delta$\\$\alpha, \rho$\\$\mathcal D^{\text{train}}_{i}, \mathcal D^{\text{val}}_{i}, \mathcal  D^{\text{test}}_{i}$ \\$S_i$ , $s_{val}$ \\$\left | \cdot \right |$ \\$\mathcal U,  U$ \\$\mathcal K, K$ \\$t, T$ \\$\boldsymbol{w}$ \\$\boldsymbol{w}_{j}$ \\$\boldsymbol{\theta}$  \\$\boldsymbol{\theta}_{i}^{t}$ \\$\boldsymbol{\hat{\theta}}_{i}^{t}$ \\$\boldsymbol{\hat{\theta}}^{t}$ \\$F(\boldsymbol{\theta}_{i})$ \\$\boldsymbol{M}$ \\$Pr$ \\$M$ \\$k$ \\$P_1^{fip}$ \\ $P_2^{fip}$\\$p$ \\$\gamma$ \\$\Delta _{i}^{t}$\end{tabular} & \begin{tabular}[c]{@{}l@{}}The public dataset and privacy dataset \\A randomized mechanism for DP \\Adjacent databases \\The parameters related to DP \\The parameters related to Rényi-DP  \\The training, validation, and test database held by the $i$-th user/client, respectively \\The validation scores and the validate function \\The cardinality of a set  \\The set of all MTs and total number of all MTs\\The set of all selected MTs and total number of selected MTs \\The index of the $t$-th communication round and the number of communication rounds \\Model parameters of all winning tickets \\Model parameters of the $j$-th winning ticket \\Model parameters of the global model  \\Model parameters of the $i$-th user/client at coomunication round $t$ \\Further pruned model parameters of the $i$-th user/client at coomunication round $t$\\The global model at coomunication round $t$ \\Global loss function from the $i$-th user  \\The mask vector of pruned model  \\The pruning degree/ratio to generate winning tickets \\The number of winning tickets \\The number of the iterations for training winning tickets  \\The initially selected WT's retention rate $P_1^{fip} = Pr$ in the fed-iterative pruning scheme \\ The further pruning degree in the fed-iterative pruning scheme \\The (averaged) final retention rate or compression ratio of the model\\The adaptive discount factor pruning in the fed-iterative pruning scheme\\Local model update from the $i$-th user at coomunication round $t$\end{tabular}  \\
\bottomrule
\end{tabular}
\end{table*}
\subsection{Federated Learning}
Consider a general FL system consisting of $U$ MTs, in which each client owns a local dataset.
Let $\mathcal D^{\text{train}}_i$, $\mathcal D^{\text{val}}_i$ and $\mathcal D^{\text{test}}_i$ denote the training dataset, validation dataset and testing dataset, held by client $i$, respectively, where $i\in \mathcal{U} = \{1, 2,\ldots, U\}$.
Formally, this FL task is formulated:
\begin{equation}
\boldsymbol{\theta}^{\star} = \mathop{\arg\min}_{\boldsymbol{\theta}}\sum_{i\in \mathcal{U}}p_{i}F(\boldsymbol{\theta}, \mathcal D^{\text{train}}_{i}),
\end{equation}
where $F(\cdot)$ is the loss function and $p_{i} = \vert \mathcal D^{\text{train}}_i\vert/\vert \mathcal D^{\text{train}}\vert\geq 0$ with $\sum_{i\in \mathcal{U}}{p_{i}}=1$; $\vert \mathcal D^{\text{train}}_{i}\vert $ is the size of training dataset $\mathcal D^{\text{train}}_i$ and $\vert \mathcal D^{\text{train}}\vert = \sum_{i\in \mathcal{U}}{\vert \mathcal D_{i}^{\text{train}}\vert}$ is the total size of training datasets, respectively.
For the $i$-th client, the updating process to learn a local model over training data $\mathcal D^{\text{train}}$ can be expressed as:
\begin{equation}
\boldsymbol{\theta}_{i}=\boldsymbol{\theta}_{i}^{t}-\alpha\nabla F(\boldsymbol{\theta}_{i}, \mathcal D^{\text{train}}_{i}).
\end{equation}

Generally, the loss function $F(\cdot)$ is given by the empirical risk and has the same expression across MTs.
Then, the $U$ associated MTs learn a global model $\boldsymbol{\theta}$ over training data $\mathcal D^{\text{train}}_{i}$, $\forall i \in \mathcal{U}$. Given the global model parameter $\boldsymbol{\theta}$ from the server by aggregation, each client $i$ can validate the model based on its validation dataset $\mathcal D^{\text{val}}_{i}$ and obtain the validation scores:
\begin{equation}
S_{i}=s_{val}(\boldsymbol{\theta}, \mathcal D^{\text{val}}_{i}),
\end{equation}
where $s_{val}$ is the validate function to calculate the number of correct predictions using the trained model $\boldsymbol{\theta}$.
\begin{figure*}[htbp]
\begin{center}
\includegraphics[width=1\textwidth]{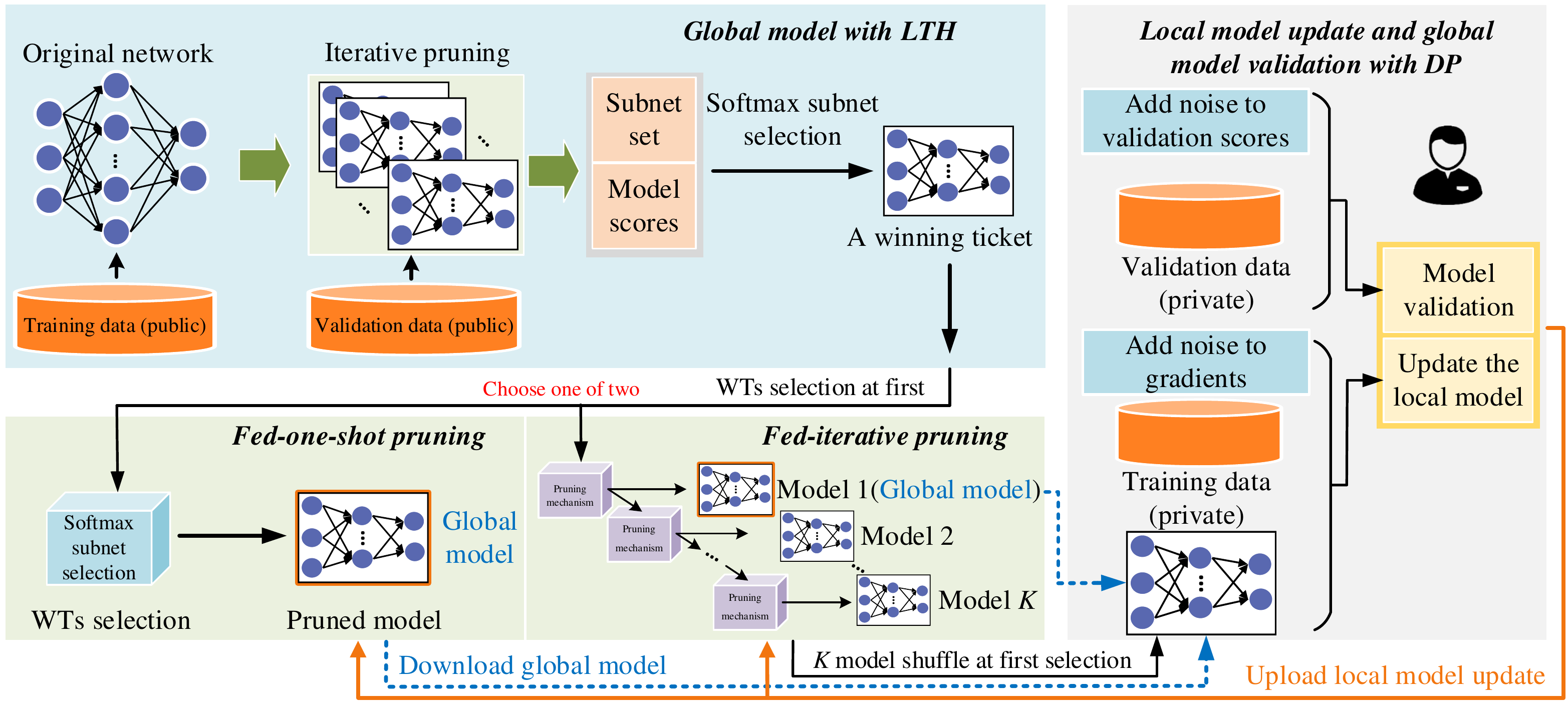}
\end{center}
\caption{ \small An overview of Fed-LTP from a client perspective with three components, where the client is the MT in the edge computing system. On the server side: 1) LTH is used to prune and train the original network to generate multiple WTs and a WT is selected as the candidate global model; 2) With the WT-broadcasting mechanism, the candidate global model is either maintained at the current degree of pruning by fed-one-shot pruning or subject to further pruning by fed-iterative pruning. On the client side: 3) The local model is trained with DP on the private data of each client and the Laplace mechanism is employed for model validation. 
}
\label{FedLTP}
\end{figure*}
\subsection{Differential Privacy}
DP \cite{Dwork2014the} is a rigorous privacy notion for measuring privacy risk. 
In this paper, we consider two relaxed versions of DP definitions: Rényi DP (RDP)~\cite{mironov2017renyi} and zero-concentrated DP (zCDP)~\cite{Bun2016Concentrated}.
\begin{definition}{\textbf{$(\epsilon, \delta)$-DP}} \cite{Dwork2014the}\textbf{.} Given privacy parameters $\epsilon > 0$ and $0 \le \delta < 1$, a randomized mechanism $\mathcal{M}$ satisfies $(\epsilon, \delta)$-DP if for any pair of adjacent datasets $\mathcal D$, $\mathcal D^{\prime}$ , and any subset of outputs $O \subseteq range(\mathcal{M})$:
\begin{equation}
\operatorname{Pr}[\mathcal{M}(\mathcal D) \in O] \leq e^{\epsilon} \operatorname{Pr}\left[\mathcal{M}\left(\mathcal D^{\prime}\right) \in O\right]+\delta.
\end{equation}
Where adjacent datasets are constructed by adding or removing any record; $(\epsilon, \delta)$-DP is $\epsilon$-DP, or pure DP when $\delta = 0$.
\end{definition}
\begin{definition}\textbf{Instance-level DP for FL} \cite{Hu2021federated}\textbf{.} A randomized algorithm $\mathcal{M}$ is $(\epsilon, \delta)$-DP if for any two adjacent datasets $\mathcal I$, $\mathcal I^{\prime}$ constructed by adding or removing any record in any client's dataset, and every possible subset of outputs $O$:
\begin{equation}
\operatorname{Pr}[\mathcal{M}(\mathcal I) \in O] \leq e^{\epsilon} \operatorname{Pr}\left[\mathcal{M}\left(\mathcal I^{\prime}\right) \in O\right]+\delta.
\end{equation}
\end{definition}
\begin{definition}{\textbf{Rényi DP}}~\cite{mironov2017renyi}\textbf{.} Given a real number $\alpha \in (1, \infty )$ and privacy parameter $\rho \ge 0$, a randomized mechanism $\mathcal{M}$ satisfies $(\alpha, \rho)$-RDP if for any two neighboring datasets $\mathcal D$, $\mathcal D'$ that differ in a single record, the Rényi $\alpha$-divergence between $\mathcal{M}(\mathcal D)$ and $\mathcal{M}(\mathcal D^{\prime})$ satisfies:
\begin{equation}
D_{\alpha}\left[\mathcal{M}(\mathcal D) \| \mathcal{M}\left(\mathcal D^{\prime}\right)\right]:=\frac{1}{\alpha-1} \log \mathbb{E}\left[\left(\frac{\mathcal{M}(\mathcal D)}{\mathcal{M}\left(\mathcal D^{\prime}\right)}\right)^{\alpha}\right] \leq \rho,
\end{equation}
where the expectation is taken over the output of $\mathcal{M}(\mathcal D^{\prime})$.
\end{definition}

To define $\rho$-zCDP, we first introduce the privacy loss random variable.
For an output $o\in \text{range}(\mathcal{M})$, the privacy loss random variable $Z$ of the mechanism $\mathcal{M}$ is defined as:
\begin{equation}
Z = \log\frac{\Pr\left[\mathcal{M}(\mathcal D)=o\right]}{\Pr\left[\mathcal{M}(\mathcal D')=o\right]}.
\end{equation}
\begin{definition}
\textbf{$\rho$-zCDP}~\cite{Bun2016Concentrated}\textbf{.}
$\rho$-zCDP imposes a bound on the moment generating function of the privacy loss $Z$ and requires it to be concentrated around zero.
Formally, it needs to satisfy:
\begin{equation}
\begin{aligned}
e^{D_{\alpha}(\mathcal{M}(\mathcal D)\|\mathcal{M}(\mathcal D'))} &= \mathbb{E}\left[e^{(\alpha-1)Z}\right] \leq e^{(\alpha-1)\alpha \rho}.
\end{aligned}
\end{equation}
\end{definition}
In this paper, we use the following zCDP composition results.

\begin{lemma}\label{lemma:zcdp}
If $\mathcal M$ satisfies $\epsilon$-differential privacy, then $\mathcal M$ satisfies $\left(\frac{1}{2}\epsilon^{2}\right)$-zCDP \cite{Bun2016Concentrated}.
\end{lemma}

\section{The proposed Approach}
In this section, we give a detailed description of Fed-LTP. The overall framework is summarized in Algorithm \ref{algo_FedLTP}, and its workflow from a client perspective is illustrated in Fig. \ref{FedLTP}. 




%



\subsection{Global Model Generation with LTH}
Inspired by DPLTM \cite{gondara2020differentially}, we use LTH to generate a candidate global model, which contains two major procedures (Algorithm \ref{algo_Global_model_with_LTH}).

\emph{1) WTs generation on public dataset with LTH.} This process is identical to LTH, with the only difference being that we use the public dataset for WTs generation.

To find a lighter-weight network with higher test accuracy, an iterative pruning method is used when generating WTs. During the $j$-th pruning, with the pruning degree to $Pr$, a mask vector ${\boldsymbol{M}}\left( {{{\boldsymbol{w}}_j}} \right)$ is set to zero if model $\boldsymbol{w}_{j}$ is pruned or one if unpruned. Let $\boldsymbol{w}_{j,j}$ denote the weight in the $j$-th layer of model $\boldsymbol{w}_{j}$. Note that for the $j$-th layer, the operation is defined as:
\begin{equation}
\boldsymbol{M}\left( \boldsymbol{w}_{j,j}  \right) = \left\{ \begin{array}{l}
\begin{array}{*{20}{c}}
1&{if\left | w  \right | > Pr \left | w _{max}  \right | }
\end{array}\\
\begin{array}{*{20}{c}}
0&{otherwise}
\end{array}
\end{array} \right.,
\label{eq:weight-based_mask}
\end{equation}
where $w  \in \boldsymbol{w}_{j,j}$ and $\left |w _{max} \right |$ denotes the largest value of $\boldsymbol{w}_{j,j}$. Therefore, the mask matrix for model $\boldsymbol{w}_{j}$ is constructed by applying \eqref{eq:weight-based_mask} to each layer.

\emph{2) WTs selection with softmax function.} Unlike the selection method in DPLTM \cite{gondara2020differentially}, for the winning ticket selection, there are $M$ alternative tickets, and each ticket is given a score $V(\boldsymbol{w}_{j} , \mathcal{D}_{{pub}})$ that equals to the correct number of samples for inference. To adjust the trade-off between the accuracy and the degree of network pruning, we use the softmax function~\cite{jang2016categorical} 
over the preference value $V(\boldsymbol{w}_{j} , \mathcal{D}_{{pub}})$ also known as the score to select WT, which ensures that all WTs are explored:
\begin{equation}
\small
{P_j} = \sigma \left( {V\left( \boldsymbol{w}_{j},\mathcal{D}_{pub} \right)} \right) = \frac{{{e^{V\left( \boldsymbol{w}_{j},\mathcal{D}_{pub} \right)}}}}{{\sum\limits_{j = 1}^M {{e^{V\left( \boldsymbol{w}_{j},\mathcal{D}_{pub} \right)}}} }},\label{eq:selectWT}
\end{equation}
where $j$ denotes one of possibly many winning tickets and $P_j$ is its corresponding probability. It is worth noting that we use LTH to generate WTs on the public dataset, and each client performs the training process with DP on their private data. Different from DPLTM \cite{gondara2020differentially}, our method does not have the privacy budget as we do not use the private data while generating the model structure.
\vspace{-0.25cm}
\begin{algorithm}
\small
\SetKwData{Left}{left}\SetKwData{This}{this}\SetKwData{Up}{up} \SetKwFunction{Union}{Union}\SetKwFunction{FindCompress}{FindCompress} \SetKwInOut{Input}{Input}\SetKwInOut{Output}{Output}
\Input{Public dataset $\mathcal{D}_{pub}$, the number of winning tickets $M$, the number of iterations for training winning tickets $k$, pruning ratio $Pr$.}
\Output{The selected winning ticket $j$.}
      \BlankLine
    \emph{\textbf{Procedure 1: WTs generation on public dataset with LTH} }\\
     \For{$j=0$ \KwTo $M$-1}{
        Randomly initialize a neural network $f\left ( \boldsymbol{w} _{0} \right ).$ \\
        Train the network for $k$ iterations on $\mathcal{D}_{pub}$ to obtain $\boldsymbol{w} _{j}$ for $j$-th winning ticket.\\
        Prune $Pr\%$ of the parameters in $\boldsymbol{w} _{j}$, creating a mask $\boldsymbol{M}$ by \eqref{eq:weight-based_mask}.\\
        Reset the remaining parameters to their values in $\boldsymbol{w} _{0}$, generating the winning ticket and a pruned model $f\left (  \boldsymbol{M} \odot \boldsymbol{w} _{0} \right ) $.
     }
     Store the pruned model, and the score $V(\boldsymbol{w}_{j}, \mathcal{D}_{pub})$.
    \BlankLine
    \emph{\textbf{Procedure 2: WTs selection with softmax function}}\\
    Select a winning ticket $j$ with probability $P_j$ by \eqref{eq:selectWT} with the score.\\
  \caption{ \small Global model generation with LTH}
  \label{algo_Global_model_with_LTH}
\end{algorithm}
\begin{algorithm*} 
\small
\SetKwData{Left}{left}\SetKwData{This}{this}\SetKwData{Up}{up} \SetKwFunction{Union}{Union}\SetKwFunction{FindCompress}{FindCompress} \SetKwInOut{Input}{Input}\SetKwInOut{Output}{Output}
\Input{The number of communication rounds $T$, local update period $\tau$, the size of selected MTs per round $K$, the clipping threshold $C$, learning rates $\eta$ and total privacy budget $\epsilon$.} 
\Output{The best global model $\boldsymbol{\hat{\theta}}$ utilizing model validation with the Laplace mechanism.}
\BlankLine 
\begin{multicols}{2}
	 \BlankLine 
	 \textbf{Server executes:}\\
       \textbf{Initialize model process:}
            Firstly, select an initially pruned models $\boldsymbol{\theta}_0$ by \emph{\textbf{Global model generation with LTH}} 
            (Algorithm \ref{algo_Global_model_with_LTH}). Secondly, get the models $(\boldsymbol{\hat{\theta}}_0 ^{t} ,...,\boldsymbol{\hat{\theta}}_{K-1}^{t})$ by \textbf{Server-side model pruning}.\\
      \For{$t=1$ \KwTo $T$}{
        Sample $K$ clients/MTs uniformly at random without replacement.\\
        Give the models by \textbf{Client model selection} to selected MTs.\\
          \For{each selected client $i$ \textbf{in parallel}}
          {$\Delta _{i}^{t} \leftarrow \textbf{ClientUpdate($\boldsymbol{\hat{\theta}}_{i}^{t}, \Delta^ {t}$)}$}
        $\boldsymbol{\hat{\theta}}^{t+1}  \gets  \boldsymbol{\hat{\theta}} ^{t} + \frac {1}{K} \sum_{i\in \mathcal{K} } \Delta _{i}^{t}$  \\
     }
     \textbf{ClientUpdate($\boldsymbol{\hat{\theta}}_{i}^{t}, \Delta^ {t}$):}
     $\boldsymbol{\theta}_i ^{t,0}  \gets  \boldsymbol{\hat{\theta}}_{i}^{t} $\\
     \For{s=0 to $\tau - 1$}{
        Compute a mini-batch stochastic gradient $\boldsymbol{g}_i ^{t,s}$\\
        $\boldsymbol{\theta}_i ^{t,s+1} \gets  \boldsymbol{\theta}_{i} ^{t,s} - \eta (\boldsymbol{g}_i ^{t,s} \times \min(1, C/\left \| \boldsymbol{g}_i ^{t,s}   \right \|_{2} )+ \boldsymbol{b}_i ^{t,s})$\\
        where $\boldsymbol{b}_i ^{t,s} \sim  \mathcal{N} (0, (\sigma ^{2}C^{2})\cdot \mathbf{\boldsymbol{\mathit{I}}}_d )$.
     }
     $\Delta _{i}^{t}  \leftarrow  \boldsymbol{\theta}_i ^{t,\tau } - \boldsymbol{\hat{\theta}}_{i}^{t}$\\
      \vspace{0.2in}
     \textbf{Return} $\Delta _{i}^{t}$\\
     \textbf{Server-side model pruning:}\\
     \If{ fed-iterative pruning}  
     {\For{the model index i, from 0 to $K - 1$}{
        Compute $\boldsymbol{\hat{\theta}}^{t}$ by \eqref{eq:iterative_pr}.
     }
     \textbf{Return} $(\boldsymbol{\hat{\theta}}_0 ^{t} ,...,\boldsymbol{\hat{\theta}}_{K-1}^{t})$\\
    }
    \If{ fed-one-shot pruning}{
        Generate $\boldsymbol{\hat{\theta}}^{t}$ by \eqref{eq:one-shot-global}.\\
    \textbf{Return} $\boldsymbol{\hat{\theta}}^{t} $\\
    }
    \textbf{Client model selection:}\\
    \If{ fed-iterative pruning}
     {Generate further pruned models $(\boldsymbol{\hat{\theta}}_0 ^{t} ,...,\boldsymbol{\hat{\theta}}_{K-1}^{t})$.\\
     Model shuffle at first selection for client $i$, and the \\models of different MTs satisfy \eqref{eq:iterative_pr_user_model}.
    }
    \If{ fed-one-shot pruning}{
        Broadcast one global model $\boldsymbol{\hat{\theta}}^{t}$ for all selected MTs.
    }
          \caption{ \small Fed-LTP algorithm}
 	 	  \label{algo_FedLTP} 
\end{multicols}
     \end{algorithm*}
\subsection{Server-side WT-broadcasting Mechanism}
We propose a server-side WT-broadcasting mechanism with two major steps (Algorithm \ref{algo_FedLTP}). It constrains the local and global models within the same model class to stabilize global model aggregation. 

\emph{1) Server-side model pruning.} We present two different pruning strategies for model pruning as below, and the final retention rate is denoted as $p$.

\textbf{Fed-one-shot pruning.} Following the convention of LTH \cite{Frankle2019the} and \cite{Luo2021scalable}, let $i \in \mathcal{K}=\{1,2,...,K\}$ denote a selected client. According to Algorithm \ref{algo_Global_model_with_LTH}, weight-based pruning is conducted when generating WTs, and the WT selected by the softmax function can be used as the global model:
\begin{equation}
\boldsymbol{\hat{\theta}}^{t} =  \boldsymbol{\theta}_0 = \boldsymbol{w}_j,
\label{eq:one-shot-global}
\end{equation}
which results in $p=1-Pr$.

\textbf{Fed-iterative pruning.} To achieve 
a better balance among test accuracy, communication cost, and resource overhead of edge devices than fed-one-shot pruning, fed-iterative pruning further prunes the global model based on LTH. Then we generate models with different pruning degrees following HeteroFL \cite{diao2020heterofl}, where local models have similar architecture
but can shrink their model size within the same global model class. 
\begin{figure}[htbp]
\begin{center}
\includegraphics[width=0.49\textwidth]{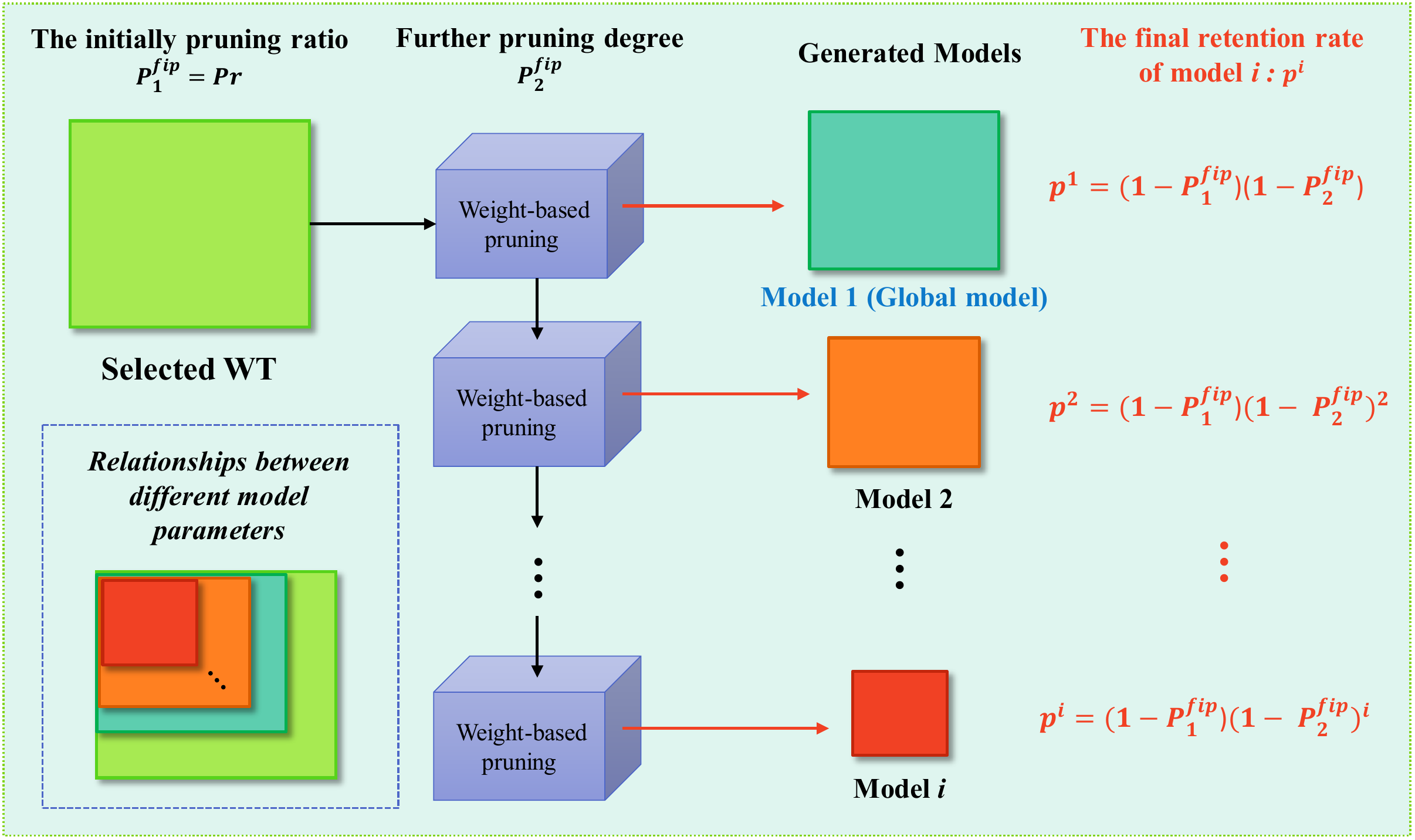}
\end{center}
\vspace{-0.25cm}
\caption{ \small An overview of fed-iterative pruning scheme.}
\label{fip}
\end{figure}

As shown in Figure \ref{fip}, $p$ is decided by two pruning factors: (i). the initially selected WT's retention rate $P_1^{fip} = Pr$; (ii). the further pruning degree $P_2^{fip}$.
In the further pruning stage, we repeatedly apply the weight-based pruning method in ~\eqref{eq:weight-based_mask} to generate heterogeneous models by iterative pruning while the number of iterations is the same as the number of MTs selected each time.
Formally,
the model parameter $\boldsymbol{\hat{\theta}}_i^{t}$ of model $i$ at the $t$-th training round is given by:
\begin{equation}\small
\left\{ \begin{array}{l}
\boldsymbol{\hat{\theta}}_i^{t} = \boldsymbol{\hat{\theta}}_{i-1}^{t} \odot  \boldsymbol{M}\left( \boldsymbol{\hat{\theta}}_{i-1}^{t} \right)\\
\boldsymbol{\hat{\theta}} _1^t = \boldsymbol{\theta} _0^t \odot \boldsymbol{M}\left( \boldsymbol{\theta} _0^t \right)
\end{array} \right.,
\label{eq:iterative_pr}
\end{equation}
where $\odot$ denotes the element-wise product.
Finally, the global model is produced by:
\begin{equation}\small
\boldsymbol{\hat{\theta}}^{t} = \boldsymbol{\theta}_0^{t} \odot \boldsymbol{M}\left( \boldsymbol{\theta}_0^{t} \right),
\end{equation}

where $\boldsymbol{\theta}_0^{t} = \boldsymbol{\theta} _0 = \boldsymbol{w}_j$ for global training round $t=0$. Note that the aggregated parameters are averaged over the unpruned parameters in each participated client, and the model structure is consistent with the client model with the highest retention rate.

\emph{2) Client model selection.} In the fed-one-shot pruning strategy, only one global model $\boldsymbol{\hat{\theta}}^{t}$ is broadcast to all MTs, whereas in the fed-iterative pruning scheme, local model $\boldsymbol{\hat{\theta}}_i$ is obtained by model shuffle (without replacement) at first selection for client $i$. It should be noted that in the subsequent training epoch, the MTs who have been selected do not participate in the model shuffle, and the models of other MTs selected at first are different from those of MTs who have been selected, that is:
\begin{equation}\label{eq:iterative_pr_user_model}
\small
\boldsymbol{M}\left( {\boldsymbol{\hat{\theta}}_i^{t}} \right) \ne {\boldsymbol{M}}\left( \boldsymbol{\hat{\theta}}_j^{t} \right), \forall i,j \in \mathcal{K} = \left\{ 1,...,K \right\}.
\end{equation}

To show the more details between the final retention rate $p$ and different client models by \eqref{eq:iterative_pr} in the fed-iterative pruning, the adaptive discount factor $\gamma ^i$ and $p ^i$ of client $i$ are given by:
\begin{equation}\label{eq:iterative_factor_user_model}
\small
\gamma ^i = (1- P_2^{fip})^i, \text{and} ~~~ p ^i = \gamma ^i(1- P_1^{fip}),
\end{equation}
where $P_1^{fip}=Pr$ and $P_2^{fip}$ also represents the degree of variation across local models.
Due to the different pruning degrees across MTs, $p$ is set to the average value of $p^i$: 
\begin{equation}\small
    p=\frac{1}{K} \sum_{i=1}^K p^i
\end{equation}

\subsection{Local Model Training with DP}
After the participated MTs download the initially pruned model from the server side, and then the MTs start the local model training with DP protection on their own private data.
Due to the initial model having been pruned in the server-side WT-broadcasting mechanism, at $t$-th communication round, each participated client $i$ performs local model iteration updates, at $s$-th local iteration step, a mini-batch stochastic gradient $\boldsymbol{g}_i ^{t,s}$ is calculated on a mini-batch private data. And then we clip $\boldsymbol{g}_i ^{t,s}$ and add DP random noise $\boldsymbol{b}_i ^{t,s}$ into it, where the noise is satisfied by the Gaussian distribution $\mathcal{N} (0, (\sigma ^{2}C^{2})\cdot \mathbf{\boldsymbol{\mathit{I}}}_d )$. Thus, the local iteration is performed as:
\begin{equation}
    \boldsymbol{\theta}_i ^{t,s+1} \gets  \boldsymbol{\theta}_{i} ^{t,s} - \eta (\boldsymbol{g}_i ^{t,s} \times \min(1, C/\left \| \boldsymbol{g}_i ^{t,s} \right \|_{2} )+ \boldsymbol{b}_i ^{t,s}). \nonumber
\end{equation}

After finishing the $\tau$ local iteration steps, we calculate the local model update:
\begin{equation}
    \Delta _{i}^{t}  \leftarrow  \boldsymbol{\theta}_i ^{t,\tau } - \boldsymbol{\hat{\theta}}_{i}^{t}. \nonumber
\end{equation}
Finally, each participated client $i$ sends the local model update $\Delta _{i}^{t}$ to the server. The more details are also summarized in lines 13-19 of Algorithm \ref{algo_FedLTP}.
Note that the local model training process in private data also means performing knowledge transfer due to the local model being trained in the public data from the server side. Below, we present some discussion about the  
knowledge transfer between public and private data.


\textbf{Knowledge transfer between public and private data.} Since the use of public data is common in DP literature~\cite{wang2020differentially,papernot2018scalable,Luo2021scalable,li2022private}, following the convention of \cite{jiang2022model,Luo2021scalable,li2022private}, we utilize the labeled public data and the computational power of the server instead of the limited resources of the edge devices. In practice, we generate several pre-trained WTs on the public data, and then save the architecture of selected WT as the global model and reinitialize the values of unpruned parameters. Note that, unlike DPLTM \cite{gondara2020differentially}, Fed-LTP does not require an additional privacy budget. 

\subsection{Model Validation with Laplace Mechanism}
After all local model updates of the participated MTs are uploaded to the server, to select the best global model and prevent over-fitting, the server calculates the validation score based on the local validation datasets after the local model update with DP in each communication round. However,
traditional schemes usually tune the hyperparameters using grid-search~\cite{Hu2021federated}, which violates the rule of the real system due to observing the testing privacy data. 
Thus, testing private data is necessary to be protected for data privacy in the real system, thereby also consuming the privacy budget. In this paper, the Laplace mechanism is adopted to achieve the DP guarantee during the validation process.

At the beginning of each communication round, the server obtains a global model $\boldsymbol{\hat{\theta}}^{t}$ by aggregation.
Then, it sends this global model to all MTs.
Each client subsequently validates the received model based on its local validation dataset to obtain scores $S_{i}^{t}$ and sends it to the server.
To protect the privacy of local validation dataset, each client needs to perturb the validation scores by:
\begin{equation}
\widetilde{S}_{i}^{t} = S_{i}^{t} + \text{Lap} \left(\Delta_{1}(s_{val})\lambda_{val}\right),
\end{equation}
where $\Delta_{1}(s_{val})$ is the DP sensitivity and $\lambda_{val}$ is the parameter for the Laplace distribution. It is clear that the maximum change of the score caused by a single sample is bounded as $\Delta_{1}(s_{val}) = 1$.
The server obtains the validation scores for the global model as:
\begin{equation}\label{eq:val_score}
\small
S^{t}=\sum_{i\in =1}^{K}\widetilde{S}_{i}^{t}=\sum_{i=1}^{K}S_{i}^{t} + \sum_{i=1}^{K}\text{Lap}( \Delta_{1}(s_{val})\lambda_{val}),
\end{equation}
after receiving all MTs' scores.
Finally, when the FL training is terminated, the server selects the model with the highest validation score as the global model, that is:
\begin{equation}
\boldsymbol{\hat{\theta}}^{f} = \boldsymbol{\hat{\theta}}^{\mathop{\arg\max}\limits_{t\in[T]}S^{t}}.
\end{equation}

From~\eqref{eq:val_score}, as the number of participants in each round gets larger, the system can obtain a more reliable validation score as the variance of the aggregated noise is smaller.


\section{Privacy Analysis}
In this section,  we provide a tight privacy analysis based on zCDP for calculating the privacy loss/budget $\epsilon$ as communication round $T$ increases.

As shown in Figure \ref{FedLTP} and Algorithm \ref{algo_FedLTP}, the privacy budget $\epsilon$ with a given $\delta$ can be divided into two parts: local model training with DP and the DP-based model validation.
\begin{theorem}
The accumulated privacy loss of the proposed algorithm after the $t$-th communication round can be expressed as
\begin{equation}\label{eq:Composition_RDP}
\small
\begin{aligned}
\epsilon^{t} &= \rho_{s}^{t} + \frac{(t+1)\alpha(\alpha-1)}{2\lambda_{val}^{2}}  + \frac{\log\left(\frac{1}{\delta}\right)-\log(\alpha)}{\alpha-1}+\log\left(1-\frac{1}{\alpha}\right),
\end{aligned}
\end{equation}
where
\begin{equation}\label{eq:Renyi_divergence}
\small
\begin{aligned}
\rho_{s}^{t} &= \frac{t+1}{(\alpha-1)}\log {\mathbb{E}_{z\sim \mu_{0}(z)}\left[\left(1-\tilde{q}+\tilde{q}\frac{ \mu_{1}(z)}{\mu_{0}(z)}\right)^{\alpha}\right]}.
\end{aligned}
\end{equation}
In \eqref{eq:Renyi_divergence}, $\mu_{0}(z) = \mathcal{N}(0, \sigma^2)$ denotes a Gaussian probability
density function (PDF),
$\mu_{1}(z) = \tilde{q}\mathcal{N}(1, \sigma^2)+(1-\tilde{q})\mathcal{N}(0, \sigma^2)$ is the PDF of a mixture of two Gaussian distributions, and $\tilde{q}$ is the sample rate of local gradient (mini-batch) in local model training.
\end{theorem}
\begin{proof}
According to references~\cite{mironov2017renyi} and~\cite{opacus}, the privacy loss can be given by:
\begin{equation}\label{eq:RDP}
\small
\begin{aligned}
\epsilon^{t} = \rho^{t} + \frac{\log\left(\frac{1}{\delta}\right)-\log(\alpha)}{\alpha-1}+\log\left(1-\frac{1}{\alpha}\right),
\end{aligned}
\end{equation}
where $\rho^{t} = \rho_{s}^{t}+\rho_{v}^{t}$ is the R{\'e}nyi $\alpha$-divergence, $\rho_{s}^{t}$ and $\rho_{v}^{t}$ is caused by the local model training with DP and the DP-based model validation, respectively. Specifically, based on~\cite{Abadi2016Deep} and~\cite{mironov2017renyi}, we can calculate $\rho_{s}^{t}$ as~\eqref{eq:Renyi_divergence}, where we use the R{\'e}nyi distance to estimate the privacy loss.
We denote $\mathcal M_{val}$ by the random mechanism used in model validation.
According to~\textbf{Lemma} \ref{lemma:zcdp} for the validation process, if $\mathcal M_{val}$ satisfies $\epsilon_{val}$-DP, it also satisfies $\frac{1}{2}\epsilon_{val}^{2}$-zCDP.
Consequently, it holds that:
\begin{equation}\label{eq:zcdp}
\small
\begin{aligned}
D_{\alpha}(\mathcal M_{val}(\mathcal{D})\|\mathcal M_{val}(\mathcal{D}')) &\leq \frac{1}{2}\alpha\epsilon_{val}^{2}(\alpha-1) = \frac{\alpha(\alpha-1)}{2\lambda_{val}^{2}}.
\end{aligned}
\end{equation}
Then, via considering (t+1) communication rounds, we can obtain 
\begin{equation}
\small
\begin{aligned}
\rho_{v}^{t} \leq \frac{(t+1)\alpha(\alpha-1)}{2\lambda_{val}^{2}}.
\end{aligned}
\end{equation}
Based on~\eqref{eq:RDP} and~\eqref{eq:zcdp}, the privacy loss is accumulated as~\eqref{eq:Composition_RDP}.
\end{proof}
During the training process, the cumulative privacy loss is updated at each epoch, and once the cumulative privacy loss exceeds the fixed privacy budget $\epsilon$, the training process is terminated.
To achieve an expected training time with a given total privacy budget, we can determine the values of hyperparameters for these schedules before training.
Note that, for the validation-based schedule, the additional privacy cost needs to be taken into account due to the access to the validation dataset.

\section{EXPERIMENTAL EVALUATION}
The goal of this section is to evaluate the performance of Fed-LTP with different final retention rates on popular benchmark datasets and compare it with other baseline methods to demonstrate the superiority of our framework.
\subsection{Experimental Setup}
\textbf{Baselines.} To evaluate the performance of Fed-LTP, we compare it with several baseline methods: 1) DP-Fed: this baseline adds instance-level DP to Fed-Avg \cite{McMahan2017Communication}; 2) Fed-SPA~\cite{Hu2021federated}: this baseline integrates random sparsification with gradient perturbation, and uses acceleration technique to improve the convergence speed.

\noindent
\textbf{Datasets and Data Partition.} 
We evaluate Fed-LTP on four datasets: MNIST \cite{lecun1998gradient},  FEMNIST \cite{caldas2018leaf}, CIFAR-10 \cite{krizhevsky2009learning}, and Fashion-MNIST \cite{xiao2017fashion}, where FEMNIST and CIFAR-10 are regarded as the public data.
Two experimental groups are presented with different private data: Fashion-MNIST
and MNIST.
Note that the implementation details and results on the \textbf{Fashion-MNIST private data} are presented in this paper. Meanwhile, detailed experimental evaluation on the MNIST private data is presented in the \textbf{Appendix} \ref{private_mnist}.
We consider two different settings for all algorithms: both identical data distributed (IID) and non-identical data distributed (non-IID) settings across federated clients, such as MTs, where non-IID means the heterogeneous data distribution of local MTs. It can make the training of the global model more difficult.
We partition the training data according to a Dirichlet distribution \textbf{Dir($\alpha$)} for each client \cite{hsu2019measuring} and generate the corresponding validation and test data for each client following the same distribution where $\alpha > 0$ is a concentration parameter controlling the uniformity among MTs. For the non-IID setting, we set $\alpha=1$, where MTs may possess samples of different numbers and classes chosen at random.


\noindent
\textbf{Implementation Details.}\ 
The number of MTs is $U=50$, and the server randomly selects a set of MTs with a sampling ratio of the MTs $q=0.1$ to participate in the training for all experiments. We set the privacy failure probability $\delta = 10^{-3}$  and the number of local iterations $\tau=300$, with a fixed clipping threshold $C=10$ and a noise multiplier $\sigma=1.4$ for all experiments. For the local optimizer on MTs, we use the momentum SGD and set the local momentum coefficient to $0.5$, while set the learning rate $\eta$ as $0.01$ with a decay rate $0.99$ for FL process. Meanwhile, the learning rate for generating the WTs is set to be $1.2\times 10^{-3} $. $P_2^{fip}$ is set to $0.1$ in the fed-iterative pruning scheme for all experiments. For the Fashion-MNIST private data, a CNN model is adopted and the number of communication rounds $T=100$ and Batch size $B=15$. More implementation details are presented in \textbf{Appendix} \ref{imple_detail}.
In addition, the privacy loss of all algorithms is calculated using the API provided in \cite{opacus}. And the communication cost of the baseline methods is calculated in the same way as in Fed-SPA \cite{Hu2021federated}. Each client in Fed-LTP also uses $p \times d \times 32 \times T \times q$ bits, where $d$ is the number of unpruned model parameters to be updated to the server.
\begin{figure*}[ht]
\centering
    \includegraphics[width=0.7\textwidth]{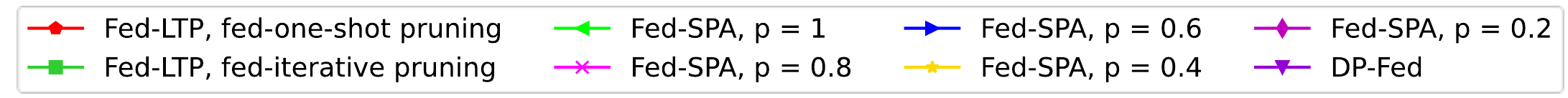}
    \subfigure[FEMNIST]{
    	\includegraphics[width=0.48\textwidth]{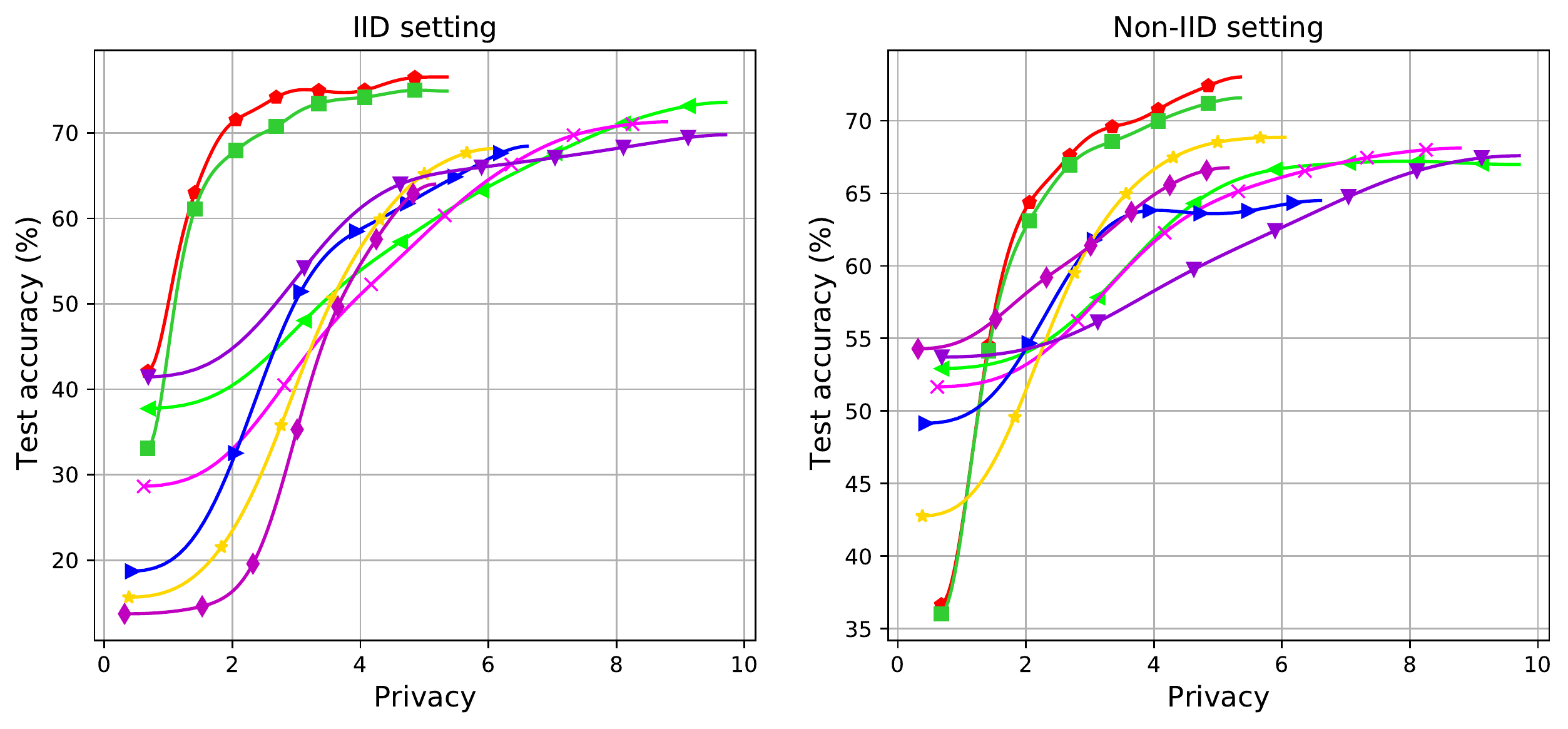}
        \label{utility_privacy_FEMNIST}
    }
    \subfigure[CIFAR-10]{
    	\includegraphics[width=0.48\textwidth]{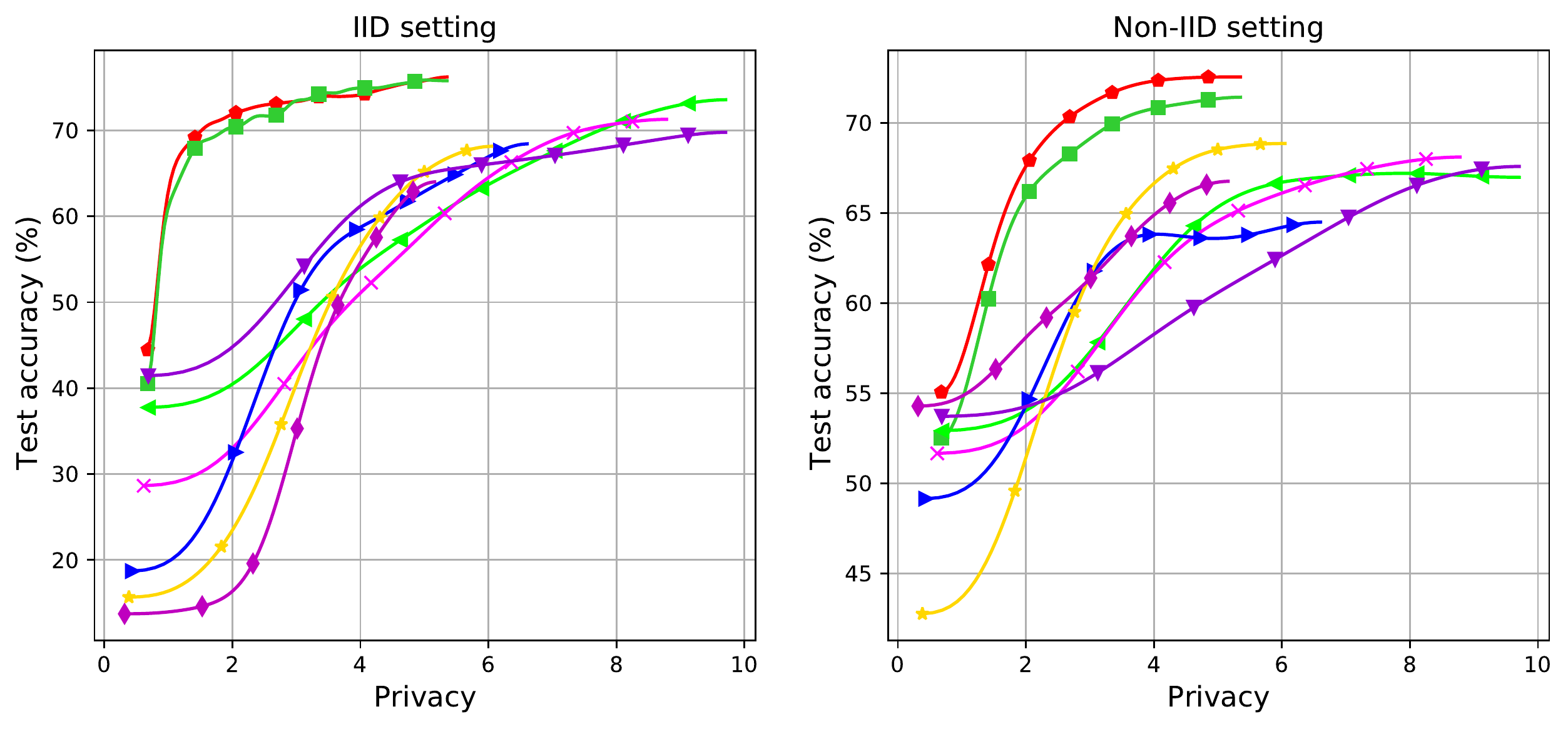}
        \label{utility_privacy_CIFAR}
    }
        \vspace{-0.4cm}
    \caption{ \small The utility-privacy trade-off of different algorithms  on two datasets in both IID and non-IID settings.}
    \vspace{-0.4cm}
    \label{utility_privacy}
\end{figure*}
\begin{figure*}[ht]
\centering

    \subfigure[FEMNIST]{
    	\includegraphics[width=0.48\textwidth]{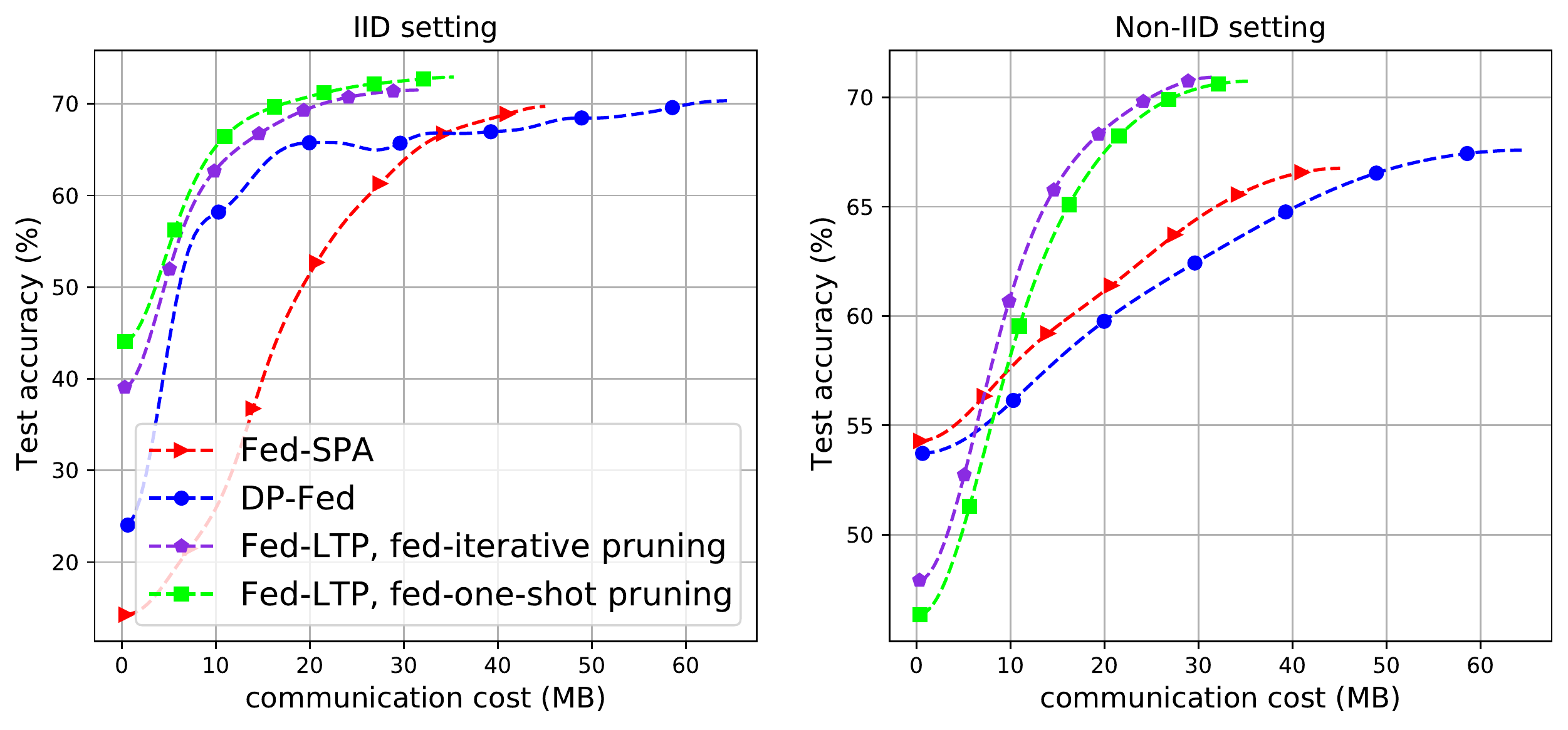}
        \label{comm_cost_FEMNIST}
    }
    \subfigure[CIFAR-10]{
    	\includegraphics[width=0.48\textwidth]{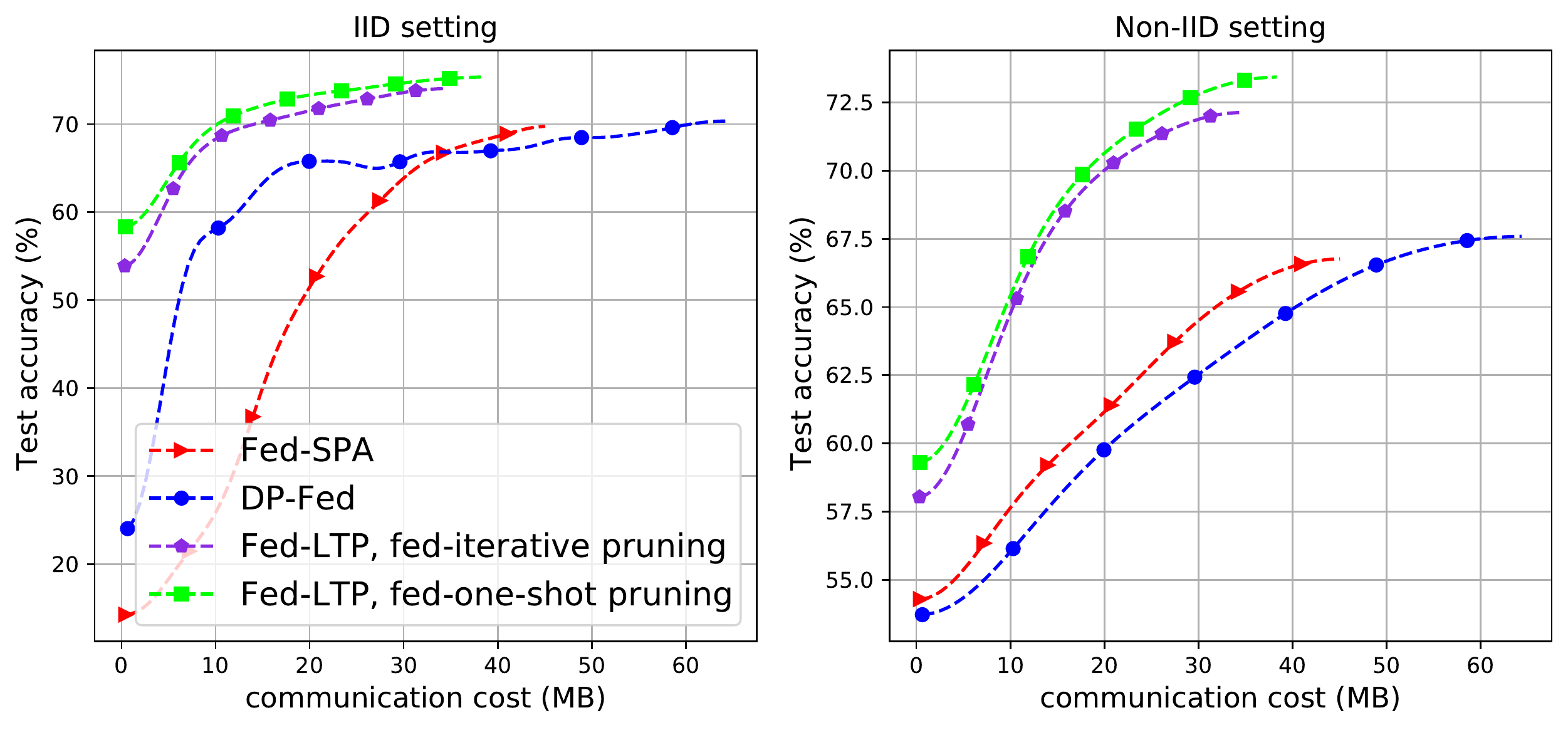}
        \label{comm_cost_CIFAR}
    }
        \vspace{-0.4cm}
    \caption{ \small Test accuracy of different algorithms with accumulated communication cost (MB)  on two datasets in both IID and non-IID settings.}
    \vspace{-0.4cm}
    \label{Accumulated_communication_cost}
\end{figure*}
\subsection{Experimental Results}
We run each experiment 3 times and report the test accuracy based on the validation datasets and the cumulative sum of upload and downstream communication costs across all rounds in each experiment. In the following, we focus on the evaluation of Fed-LTP from various aspects. 




\begin{table}
\centering
\scriptsize
\caption{ \small The model compression ratio stored on MTs $R_{loc}$ of Fed-LTP under fed-iterative pruning and baseline methods at different final retention rates $p$ on the Fashion-MNIST private data.}
\label{final_retention_rate}
\renewcommand{\arraystretch}{1}
\resizebox{1\linewidth}{!}{
\begin{tabular}{cccccccc} 
\toprule
Methods                  & \multicolumn{1}{c}{Public data}                                             & \multicolumn{1}{c}{ $p$ (Average)} & \multicolumn{5}{c}{$R_{loc}$ (Across selected MTs)}  \\ 
\midrule
\multirow{8}{*}{Fed-LTP} & \multirow{4}{*}{FEMNIST} & 0.20                                                                                                & 0.44  & 0.27  & 0.16  & 0.10  & 0.06                          \\
                         &                                                                          & 0.29                                                                                                & 0.52  & 0.36  & 0.25  & 0.18  & 0.12                          \\
                         &                                                                          & 0.40                                                                                                & 0.59  & 0.47  & 0.38  & 0.30  & 0.24                          \\
                         &                                                                          & 0.54                                                                                              & 0.66  & 0.60  & 0.54  & 0.48  & 0.44                          \\ 
\cmidrule{2-8}
                         & \multirow{4}{*}{CIFAR-10}                                                 & 0.28                                                                                               & 0.60 & 0.36 & 0.22 & 0.13 & 0.08                         \\
                         &                                                                          & 0.39                                                                                               & 0.70 & 0.49 & 0.34 & 0.24 & 0.17                         \\
                         &                                                                          & 0.53                                                                                               & 0.79 & 0.64 & 0.51 & 0.41 & 0.33                         \\
                         &                                                                          & 0.73                                                                                               & 0.89 & 0.80 & 0.72 & 0.65 & 0.59                         \\ 
\midrule
Fed-SPA                  & --                                                                      & All                                                                                                & 1.00    & 1.00    & 1.00    & 1.00    & 1.00                            \\
\midrule
DP-Fed                   & --                                                                      & 1.00                                                                                                 & 1.00    & 1.00    & 1.00    & 1.00    & 1.00                            \\
\bottomrule
\end{tabular}
}
\end{table}

\textbf{1) Efficient computation and memory footprint on MTs.} Suppose that the model parameter is represented by a 32-bit floating number and the model compression ratio stored on MTs is $R_{loc}$, which is also regarded as the key parameter that can determine the size of computation overhead and memory footprint on MTs. Since the computational and memory footprint overhead required by MTs is proportional to the size of the training models, the resource constraint for MTs is alleviated when the model size is  reduced.
In particular, the local model is trained with the sparse-to-sparse technique in Fed-LTP, while the baseline methods Fed-SPA and DP-Fed use the dense-to-sparse and dense-to-dense training, respectively. In practice, Fed-SPA only reduces the uploading communication cost without reducing the model size, while DP-Fed trains the network without any model compression or pruning. Therefore, $p \le R_{loc}=1.00$ for all baselines, and $p = R_{loc}<1.00$ for Fed-LTP. 
Furthermore, considering the resource heterogeneity across MTs, we propose the fed-iterative pruning strategy, where the values of $R_{loc}$ among the selected MTs are different and small compared with baselines as shown in Table \ref{final_retention_rate}, while for the fed-one-shot pruning strategy, the values of $R_{loc}$ are the same as the value of $p$. Consequently, the resource overhead of fed-iterative pruning is less than that of fed-one-shot pruning when choosing the same WT as the candidate global model.
\begin{figure*}[htbp]
    \centering
    \subfigure[Fed-iterative pruning]{
        \includegraphics[width=0.48\textwidth]{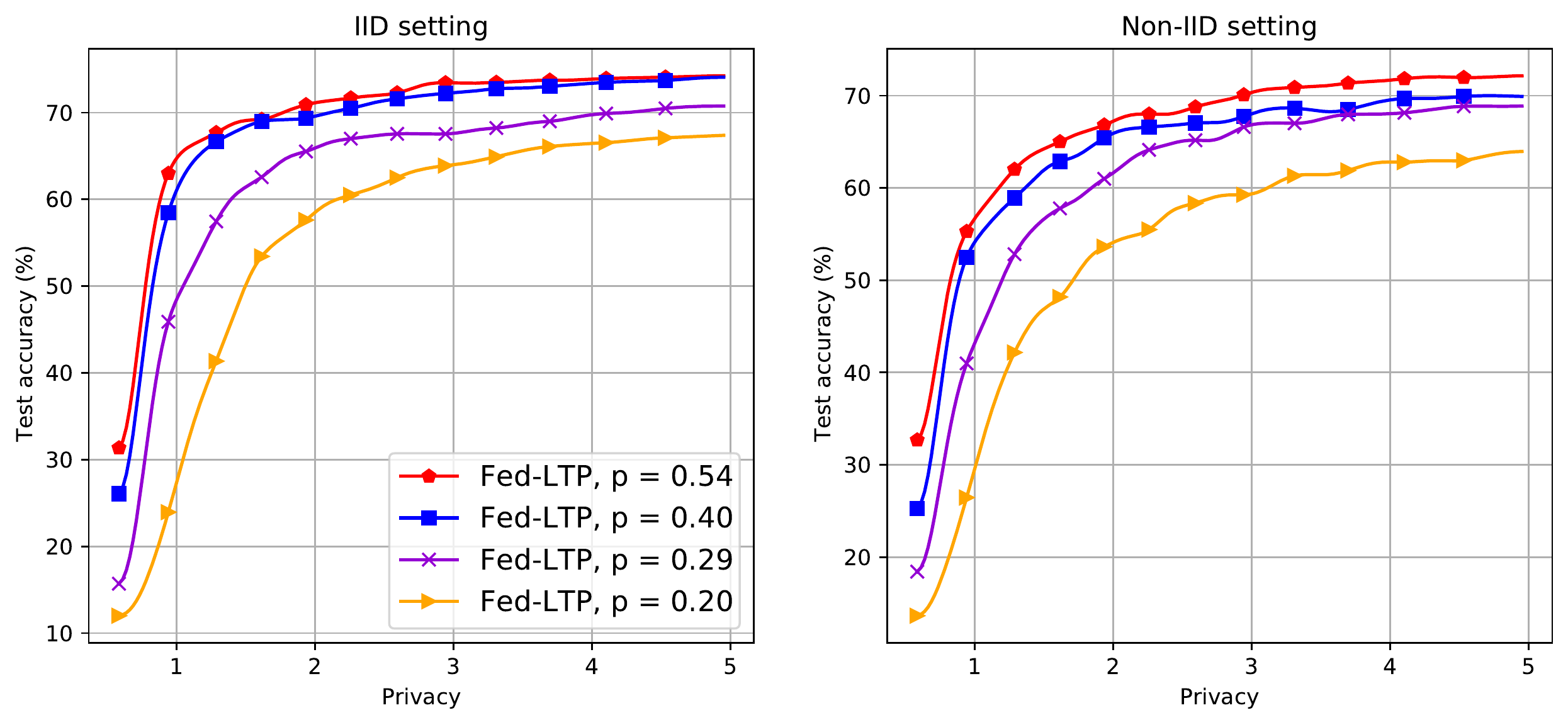}
        \label{ablation_p_fip}
    }
    \subfigure[Fed-one-shot pruning]{
    	\includegraphics[width=0.48\textwidth]{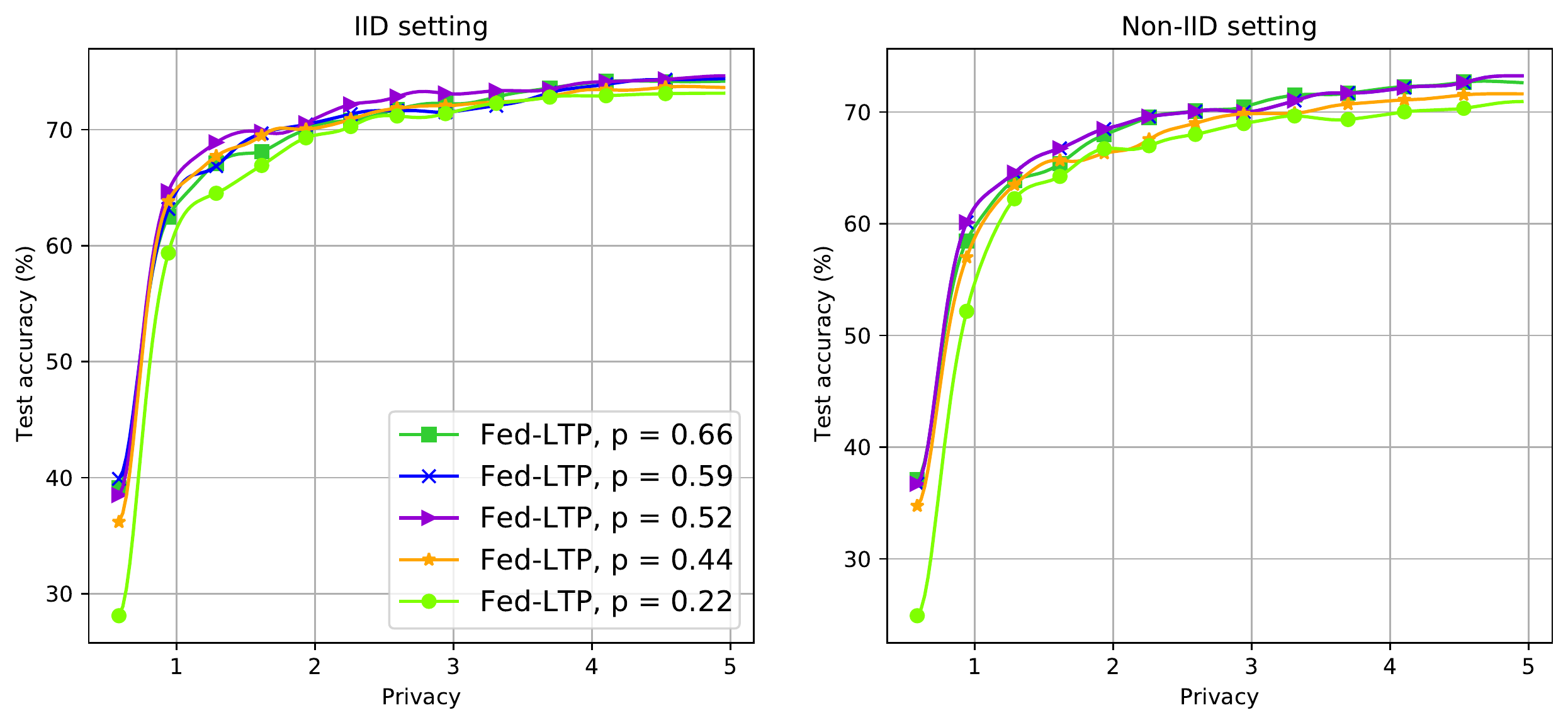}
        \label{ablation_p_fop}
    }
        \vspace{-0.4cm}
    \caption{ \small Impact of pruning: Test accuracy of the global model with the accumulated privacy loss $\epsilon$ at different retention rates $p$ using fed-iterative pruning and fed-one-shot pruning strategies in both IID and non-IID settings.}
    \label{ablation_p}
\end{figure*}

\textbf{2) Better utility-privacy trade-off.} \
We compare the best testing accuracy of all algorithms under the same privacy budget, named the utility-privacy trade-off.
In Fig. \ref{utility_privacy}, Fed-LTP achieves better utility-privacy trade-off with two different pruning strategies than baselines on two datasets: FEMNIST and CIFAR-10 in both IID and non-IID settings. Specifically, with the same privacy loss, Fed-LTP has better test accuracy than baselines. For instance, when $\epsilon=4.0 $ in Fig. \ref{utility_privacy_FEMNIST}, Fed-LTP increases the test accuracy by around $21\%$ and $9\%$ compared with baselines on FEMNIST in the IID and non-IID settings, respectively. Meanwhile, the convergence speed of the model training in Fed-LTP is higher than that in baselines. Therefore, Fed-LTP achieves a better utility-privacy trade-off, which means better model performance and stricter privacy guarantees.

\textbf{3) Efficient communication.} For each algorithm with its optimal retention rate, Fig. \ref{Accumulated_communication_cost} shows its testing accuracy with respect to the cumulative sum of upload and downstream communication costs. On FEMNIST and CIFAR-10, the algorithm settings are: Fed-LTP (fed-iterative pruning, $p=0.40$ and $p=0.39$, respectively; fed-one shot pruning, $p=0.40$ and $p=0.30$, respectively), Fed-SPA ($p=0.6$), and DP-Fed ($p=1$). 
It is clear that for different settings and datasets, the cumulative total communication cost of Fed-LTP is lower than those of the baseline methods while having better convergence speed and test accuracy.
As shown in Fig. \ref{comm_cost_FEMNIST} and Fig. \ref{comm_cost_CIFAR}, with the same communication cost, Fed-LTP under two pruning strategies always achieve better accuracy than baselines due to the use of LTH and further pruning the global model on the server side. 
Consequently, it is clear that Fed-LTP with two different pruning strategies are more communication-efficient than baselines, and the fed-one-shot pruning strategy is more communication-efficient than the fed-iterative pruning strategy. For instance, to achieve a target accuracy $70\%$ on FEMNIST in Fig. \ref{comm_cost_FEMNIST}, in the IID setting, these two pruning strategies produce the communication cost of around $22$MB and $13$MB, respectively; in the non-IID setting, they yield the communication cost of around $25$MB and $18$MB.
\begin{figure}[htbp]
    \centering
    \includegraphics[width=0.48\textwidth]{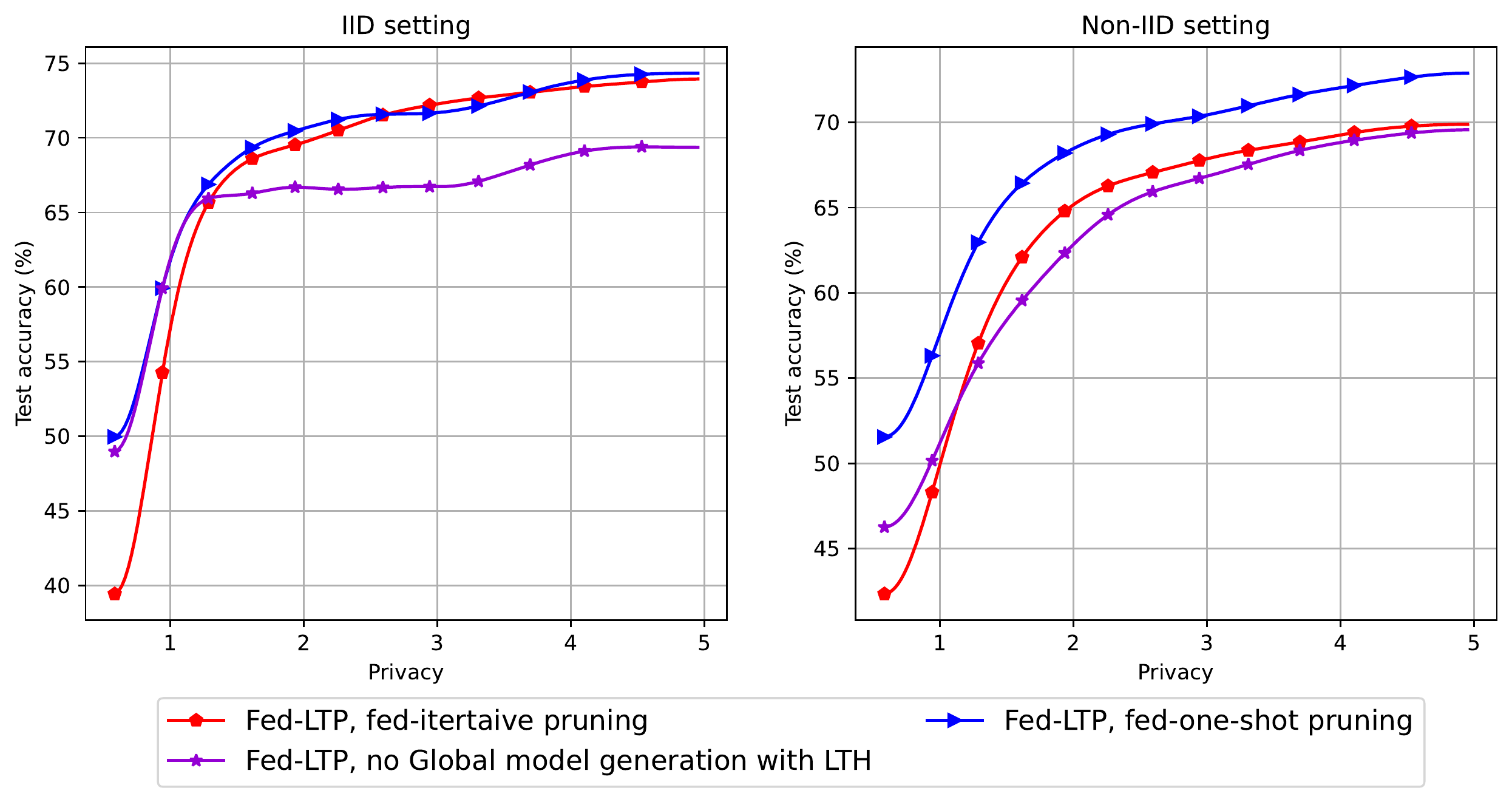}
    \caption{ \small Effect of global model generation with LTH: Test accuracy of the global model with accumulated privacy loss $\epsilon$ under different scenarios with or without a global model generation with LTH module. All scenarios are trained with fed-iterative pruning and fed-one-shot pruning strategies in both IID and non-IID settings.}
    \label{ablation_Global_model_with_LTH}
\end{figure}
This observation confirms that client heterogeneity and downstream cost are two meaningful factors affecting the model performance and communication efficiency, respectively.

\subsection{Discussion of the Pruning Schemes and Retention Rates}
In this section, the impact of pruning in Fed-LTP is investigated with various retention rates and pruning schemes for IID and non-IID data from FEMNIST.

\textbf{Impact of pruning (the final retention rate $p$).} In Fig. \ref{ablation_p}, as the retention rate $p$ decreases, the test accuracy tends to decrease. This can be analyzed from the perspective of parameter sharing, as the loss of model information perceived by the server is evident when the retention rate $p$ is small, causing significant errors in the training process. For instance, in Fig. \ref{ablation_p_fip}, as the retention rate $p$ is set to $0.54$, $0.40$, $0.29$, and $0.20$ with the fed-iterative pruning strategy in the IID setting, test accuracy decreases to around $72\%$, $70\%$, $66\%$ and $58\%$, respectively. 
Furthermore, the decrease in test accuracy is particularly evident with the fed-iterative pruning than with the fed-one-shot pruning (e.g., $p=0.20$ and $p=0.22$ in both IID and Non-IID settings, respectively). 

\textbf{Impact of the pruning scheme.} The difference in performance between fed-iterative pruning and fed-one-shot pruning strategies can be observed from Fig. \ref{utility_privacy} and Fig. \ref{ablation_p}, which is clear that the fed-iterative pruning achieves a performance similar to fed-one-shot pruning while reducing the resource overhead of MTs and communication cost. That means a better balance between performance, computation overhead, and communication cost can be achieved by reducing the complexity of the local models. However, especially in the non-IID setting as shown in Fig. \ref{ablation_p_fip}, the performance is worse than fed-one-shot pruning due to the dual effects of both model and data heterogeneity as the (averaged) final retention rate $p$ decreases. Our current experimental study on the effect of data heterogeneity, and the effect of client heterogeneity across MTs is an interesting extension for future work.
\subsection{Ablation Study}

In this section, we present the ablation study of Fed-LTP. The purpose is to investigate the specific role and effect of  a certain component or hyper-parameter in Fed-LTP, by fixing others to their default values. Ablation experiments are conducted on FEMNIST.

\textbf{Effect of global model generation with LTH.} LTH is used to generate a unified sparse structure of the global model while ensuring better model performance on the server side. 
To validate the effect of global model generation with LTH, we conduct experiments on Fed-LTP with two pruning strategies, and Fed-LTP without the global model generation for comparison, where the unpruned network is used. As in the case of LTH in a centralized learning scenario~\cite{Frankle2019the,frankle2019stabilizing,frankle2020linear}, Fig. \ref{ablation_Global_model_with_LTH} shows that the global model with the LTH module can improve the performance under the two pruning strategies: the test accuracies of Fed-LTP with fed-iterative pruning are improved by around $5\%$ and $1\%$ in the IID and non-IID settings, respectively. Therefore, with the fed-one-shot pruning, the performance gains are around $5\%$ and $4\%$, respectively.

\begin{figure}[htbp]
    \centering
    \includegraphics[width=0.48\textwidth]{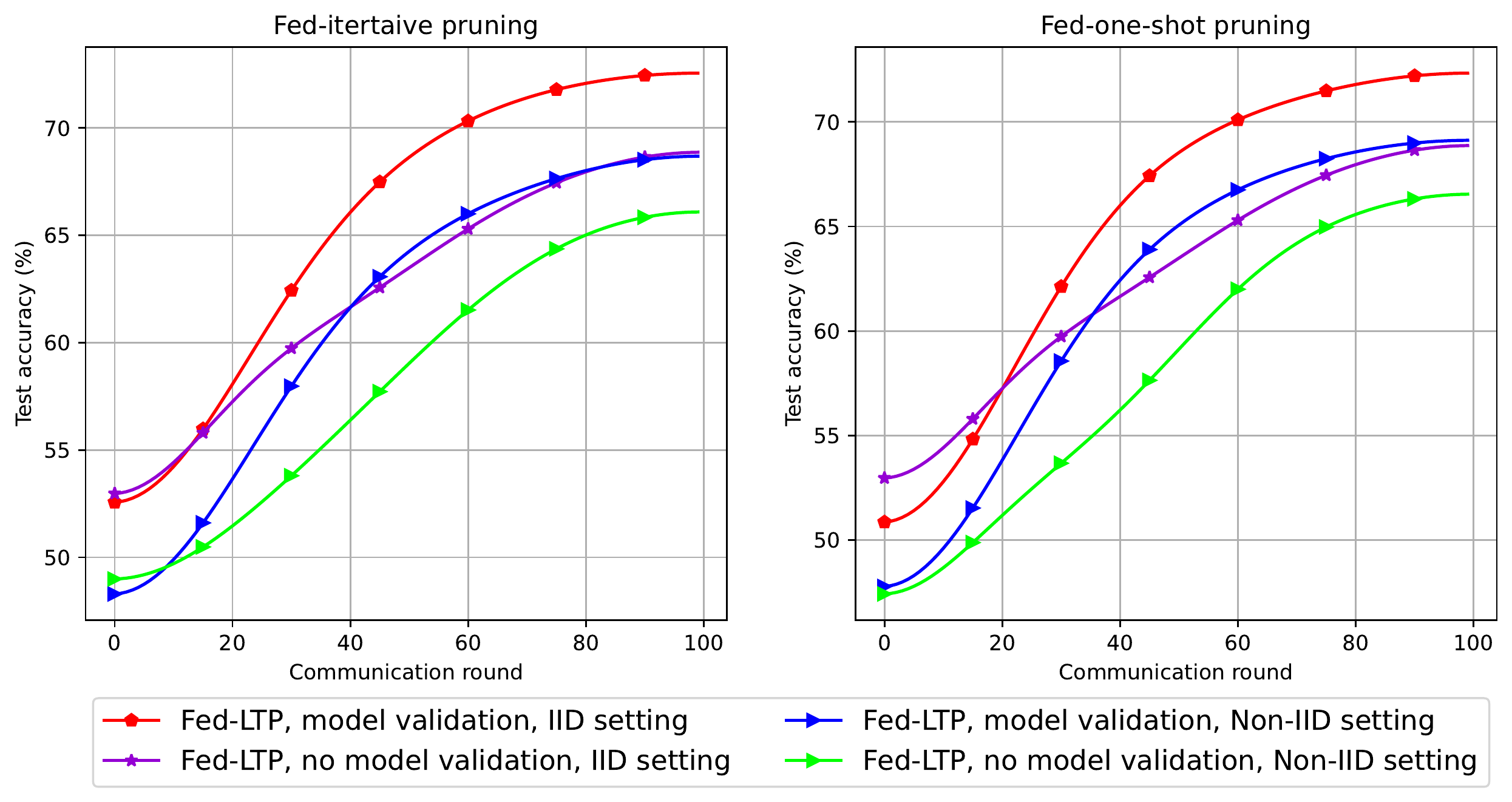}
    \vspace{-0.3cm}
    \caption{ \small Effect of model validation with the Laplace mechanism: Test accuracy of the global model with accumulated privacy loss $\epsilon$ under different scenarios with or without a model validation with the Laplace mechanism module, which was trained in both IID and non-IID settings using fed-iterative pruning and fed-one-shot pruning.}
    \label{ablation_model_selection}
\end{figure}
\begin{figure}[htbp]
    \centering
    \includegraphics[width=0.35\textwidth]{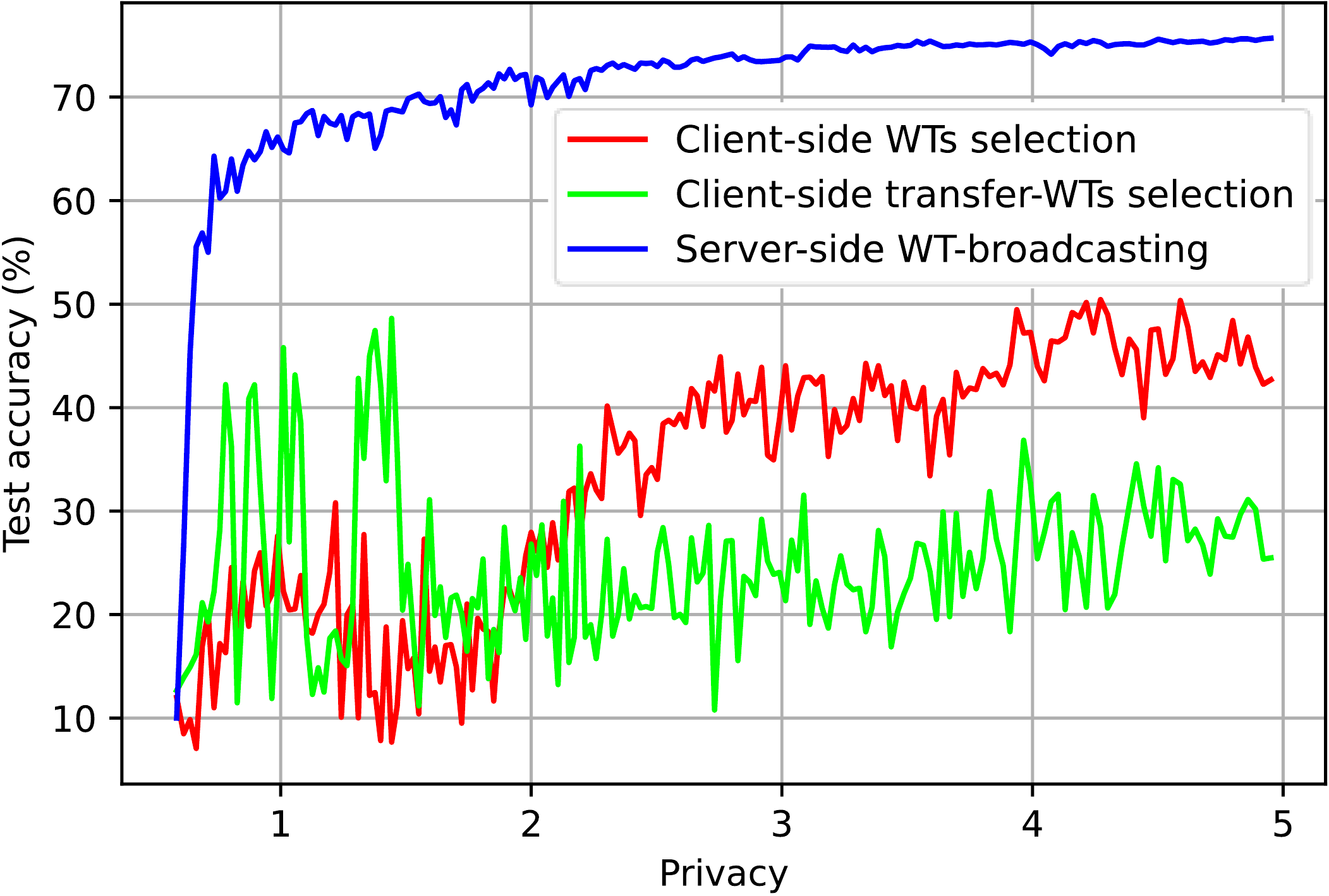}
    \vspace{-0.3cm}
    \caption{ \small Effect of  server-side WT-broadcasting mechanism: Test accuracy of the global model with accumulated privacy loss $\epsilon$ at the different schemes in the IID setting. }
    \vspace{-0.3cm}
    \label{ablation_server_side}
\end{figure}
\textbf{Effect of server-side WT-broadcasting mechanism.} To verify the effect of the server-side WT-broadcasting mechanism, experiments are conducted on three schemes: 1) client-side WTs selection, where MTs generate and save WTs with local private data on the client side; 2) client-side transfer-WTs selection, where MTs only need to select and train a pretrained model created by the server side; 3) the server-side WT-broadcasting mechanism. 
The results in Fig. \ref{ablation_server_side} show that the client-side WTs selection and client-side transfer-WTs selection schemes are inferior to the server-side WT-broadcasting mechanism with more volatile performance, due to the large variations in the structure of WTs generated by different MTs and the biased data distribution across MTs.

\textbf{Effect of model validation with the Laplace mechanism.} When the global model validation is not used on the server side, the global model in the final training process can be regarded as the final model. Therefore, the performance of the trained model usually deteriorates with a large number of rounds especially in the DP setting. This can result in a worse performance of the final model than using the global model validation.
As shown in Fig \ref{ablation_model_selection}, it is clear that the test accuracy of Fed-LTP with model validation is better in both settings and with different pruning schemes.

\textbf{Effect of  privacy loss computing method (zCDP).}
As the privacy analysis introduced in Section V, the privacy loss with zCDP can increase the level of privacy protection during the communication round. In Fig. \ref{ablation_zCDP}, the results are generally worse in terms of both the convergence speed and the best accuracy of the global model than Fed-LTP without the privacy loss with zCDP. Therefore, zCDP is shown to be beneficial to the model utility and privacy guarantee in both settings and with different pruning schemes.

\begin{figure}[htbp]
    \centering
    \includegraphics[width=0.48\textwidth]{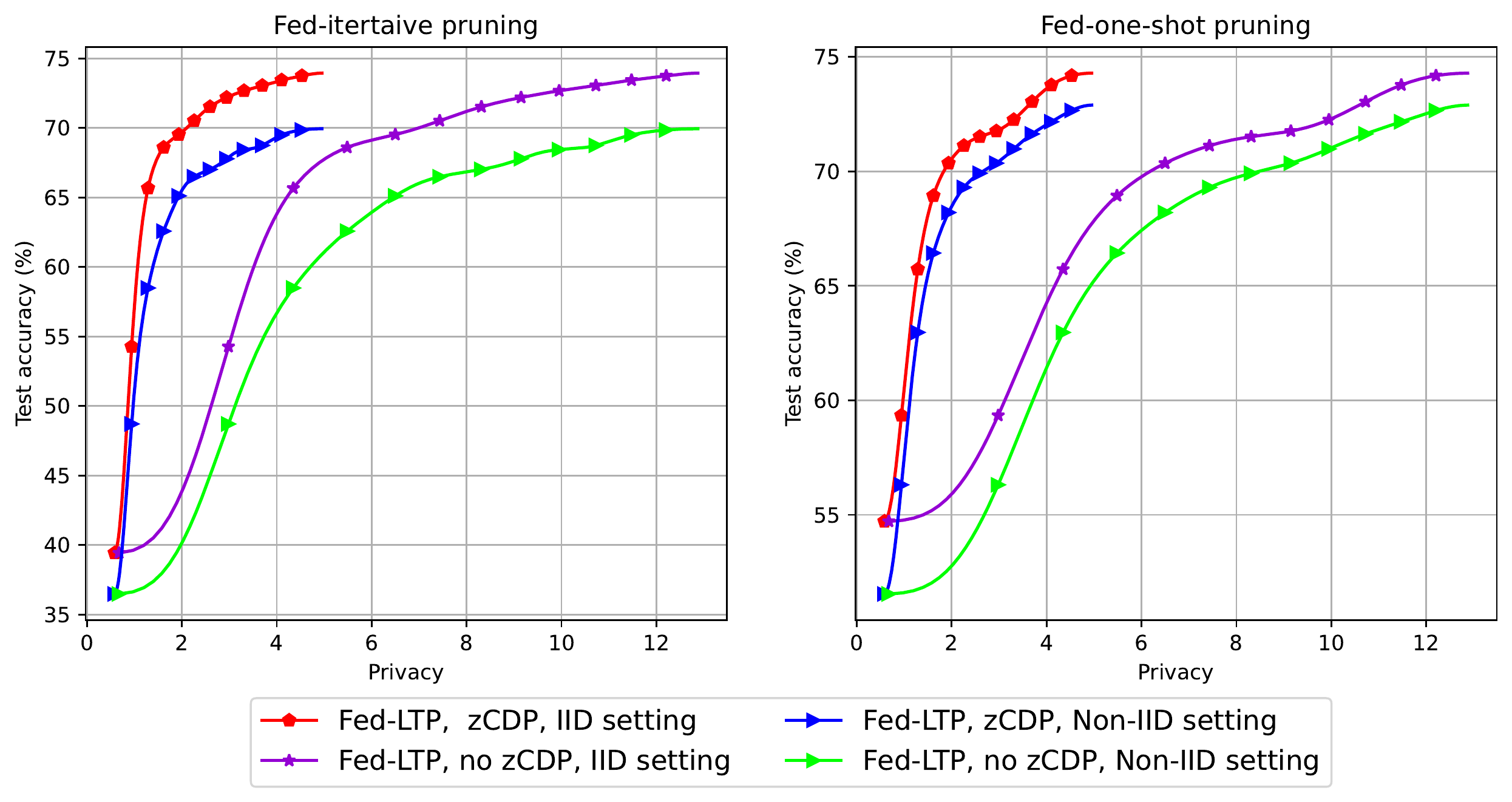}
    \vspace{-0.3cm}
    \caption{ \small Effect of  privacy loss computing method (zCDP): Test accuracy of the global model with accumulated privacy loss $\epsilon$ under different scenarios with or without a zCDP module, which was trained in both IID and non-IID settings using fed-iterative pruning and fed-one-shot pruning strategies. }
    \label{ablation_zCDP}
\end{figure}

\section{CONCLUSIONS}
In this paper, we design a privacy-preserving algorithm in FL (Fed-LTP) that can properly balance computation, memory efficiency of edge devices, and communication efficiency with improved model utility.
It contains a pre-trained model for exploring the sparse network structure and a differentially private global model validation mechanism to ensure the quality of the selected model against over-fitting. 
Meanwhile, we present the privacy analysis combining the privacy costs of model training and validation, and adopt the sparse-to-sparse training to save the limited resources of edge devices. Furthermore, the proposed noise-adding approach and the lightweight model can result in a better balance between the privacy budget and model performance. Finally, extensive experiments are conducted to verify the effectiveness and superiority of the proposed algorithm compared with SOTA methods. For future work, we will further investigate the effect of the iterative nature of the pruning method across clients/MTs with non-IID datasets. In addition, the ability to validate generalization guarantees on non-IID datasets also needs further exploration.
\ifCLASSOPTIONcaptionsoff
  \newpage
\fi



%


\bibliographystyle{IEEEtran}
\bibliography{sample-base}
%





\clearpage

\appendix


\subsection{Notation and high parameters}
In Table \ref{notion}, we present the notion and high parameters used in this paper.

\subsection{Implementation details}\label{imple_detail}
\textbf{Models.} For the MNIST private dataset, the CNN model consists of two 5 $\times$ 5 convolution layers with the ReLu activation function (the first with 10 filters, the second with 20 filters, each followed with 2 $\times$ 2 max pooling), a fully connected layer with 320 units and the ReLu activation function, and a final softmax output layer, referred to as \emph{model 1}. For the Fashion-MNIST private dataset, a CNN model is adopted, which is identical to \emph{model 1} except that the convolutional layer size is 3 $\times$ 3 with the ReLu activation function (the first with 32 filters, the second with 64 filters) and a fully connected layer has 512 units, referred to as \emph{model 2}. In addition, the number of communication rounds $T=50$ and Batch size $B=10$ for \emph{model 1}, and $T=100$ and $B=15$ for \emph{model 2}. 

\noindent
\textbf{Datasets.} 
There are 60K training examples and 10K testing examples on the datasets: MNIST, FEMNIST, and Fashion-MNIST. Specifically, The MNIST and Fashion-MNIST datasets consist of 10 classes of 28 $\times$ 28 handwritten digit images and grayscale images, respectively. The FEMNIST dataset is built by partitioning the data in Extended MNIST based on the writer of the digit/character, which consists of 62 classes. The CIFAR-10 dataset consists of 10 classes of 32 $\times$ 32 images. There are 50K training examples and 10K testing examples in the dataset.

\subsection{The detail of experimental results on the Fashion-MNIST private data}
\begin{table*}
\scriptsize
\caption{The results of Fed-LTP and baseline methods on two public datasets in both IID and non-IID settings under the Fashion-MNIST private dataset. Note that we have the same communication cost and privacy budget $\epsilon$ in both IID and non-IID settings due to using the same model structure. Meanwhile, “Acc” and “Comm”  in Table \ref{all_methods_results} is regarded as the best testing accuracy of the global model and the cumulative sum of upload and downstream communication costs across all rounds, respectively. The compression ratio for Fed-SPA is $p$, which can also be viewed as the final retention rate. For DP-Fed, $p$ = 1.0 without acceleration technique. The privacy loss $\epsilon$ for all algorithms is accumulated across communication rounds \cite{opacus}. For Fed-LTP, $\epsilon$ is accumulated by \eqref{eq:Composition_RDP},  thus it is independent of the final retention rate $p$.}
\label{all_methods_results}
\renewcommand{\arraystretch}{0.8}
\centering
\resizebox{0.9\linewidth}{!}{
\begin{tabular}{ccccc|cccc} 
\toprule
\multirow{3}{*}{\begin{tabular}[c]{@{}c@{}}Methods\\ (The final retention rate: $p$)\end{tabular}}        & \multicolumn{4}{c|}{FEMNIST}                                                     & \multicolumn{4}{c}{CIFAR10}                                                                                                  \\ 
\cmidrule{2-9}
                                & \multicolumn{2}{c}{Acc}         & \multirow{2}{*}{Comm(MB)} & \multirow{2}{*}{$\epsilon$} & \multicolumn{2}{c}{Acc}                                                    & \multirow{2}{*}{Comm(MB)} & \multirow{2}{*}{$\epsilon$}  \\ 
\cmidrule{2-3}\cmidrule{6-7}
                                & IID            & Non-IID        &                           &                    & IID            & Non-IID                                                   &                           &                     \\ 
\midrule
Fed-LTP (fed-iterative pruning) & 72.25          & \textbf{71.23} & \textbf{ 31.66 }          & 5.35               & 75.16          & 72.57                                                     & \textbf{34.38}                     & 5.35                \\
Fed-LTP (fed-one-shot pruning)  & 73.65          & 70.86          & 35.30                     & 5.35               & \textbf{75.87}        & \begin{tabular}[c]{@{}c@{}}\textbf{73.47} \\\end{tabular} & 38.34           & 5.35                \\ 
\midrule
Fed-SPA, $p$ = 1                  & \textbf{74.40} & 68.40          & 64.36                     & 9.71               & 74.40 & 68.40                                                     & 64.36                     & 9.71                \\
Fed-SPA, $p$ = 0.8                & 72.21          & 69.09          & 57.92                     & ~8.78              & 72.21          & 69.09                                                     & 57.92                     & ~8.78               \\
Fed-SPA, $p$ = 0.6                & 70.45          & 65.40          & 51.49                     & 7.77               & 70.45          & 65.40                                                     & 51.49                     & 7.77                \\
Fed-SPA, $p$ = 0.4                & 69.21          & 69.23          & 45.05                     & 6.60               & 69.21          & 69.23                                                     & 45.05                     & 6.60                \\
Fed-SPA, $p$ = 0.2                & 65.26          & 68.96          & 38.62                     & \textbf{5.16}      & 65.26          & 68.96                                                     & 38.62                     & \textbf{5.16}       \\
DP-Fed                          & 70.41          & 68.89          & 64.36                     & 9.71               & 70.41          & 68.89                                                     & 64.36                     & 9.71                \\
\bottomrule
\end{tabular}}
\end{table*}
We compare the best testing accuracy of all algorithms under the same privacy budget, named the utility-privacy trade-off, which is shown in Figure \ref{utility_privacy}. Furthermore, we list the final results in Table \ref{all_methods_results}, which summarizes the final results of DP-Fed, Fed-SPA, and Fed-LTP after $T$ rounds on Fashion-MNIST private dataset in both IID and non-IID settings. 
The cost of the baseline methods is calculated in the same way as in Fed-SPA \cite{Hu2021federated}. Each client in Fed-LTP uses $p \times d \times 32 \times T \times q$ bits, where $d$ is the number of unpruned model parameters to be updated to the server. 

It is very clear that there are three main evaluation indicators and that our algorithm achieves a good balance between them. In order to facilitate discussion and comparative analysis, the following sections are divided into separate discussions and analyses of other indicators under a fixed indicator. We can see that the performance in the non-IID setting is generally worse than that in the IID setting due to the data heterogeneity across the federated clients.
Moremore, Fed-LTP under the two pruning strategies also achieves a better trade-off between accuracy and communication cost, in comparison to the two baseline methods DP-Fed and Fed-SPA with different final retention rates $p$. Specifically, Fed-SPA features better test accuracy but also higher communication and privacy costs when $p$ is large. Meanwhile, DP-Fed is generally worse than Fed-LTP in terms of test accuracy and communication cost.
\subsection{The detail of experimental results on the MNIST private data}\label{private_mnist}
The results are the same as that on Fashion-MNIST private dataset. In particular, we list the evaluation from various aspects.
\subsubsection{Efficient computation and memory footprint of edge devices.}
The computational and memory footprint overhead required by edge devices is proportional to the size of the training models. On the MNIST private dataset, the training model size is the same as the model size on the Fashion-MNIST private dataset. Therefore, the results of computation and memory footprint overheads are also the same as the results on the Fashion-MNIST private dataset as shown in Table \ref{final_retention_rate}. 
\subsubsection{Better utility-privacy trade-off.}

In Figure \ref{mnist_utility_privacy}, we present the best testing accuracy of all algorithms under the same privacy budget $\epsilon$. It is clearly seen that Fed-LTP generates better testing accuracy under a smaller privacy budget. That means our algorithm achieves a better utility-privacy trade-off, which is same as the results on Fashion-MNIST private dataset.

\subsubsection{Efficient communication.}

In Figure \ref{mnist_communication_cost}, we present the best testing accuracy of all algorithms under the same communication cost. It is clearly seen that Fed-LTP generates better testing accuracy at a smaller cost. That means our algorithm achieves a better utility-communication trade-off, which is same as the results on Fashion-MNIST private dataset.

\begin{figure*}[ht]
\centering
    \includegraphics[width=0.6\textwidth]{content/iid_utility_privacy_legend.pdf}
    \subfigure[IID setting]{
        \includegraphics[width=0.48\textwidth]{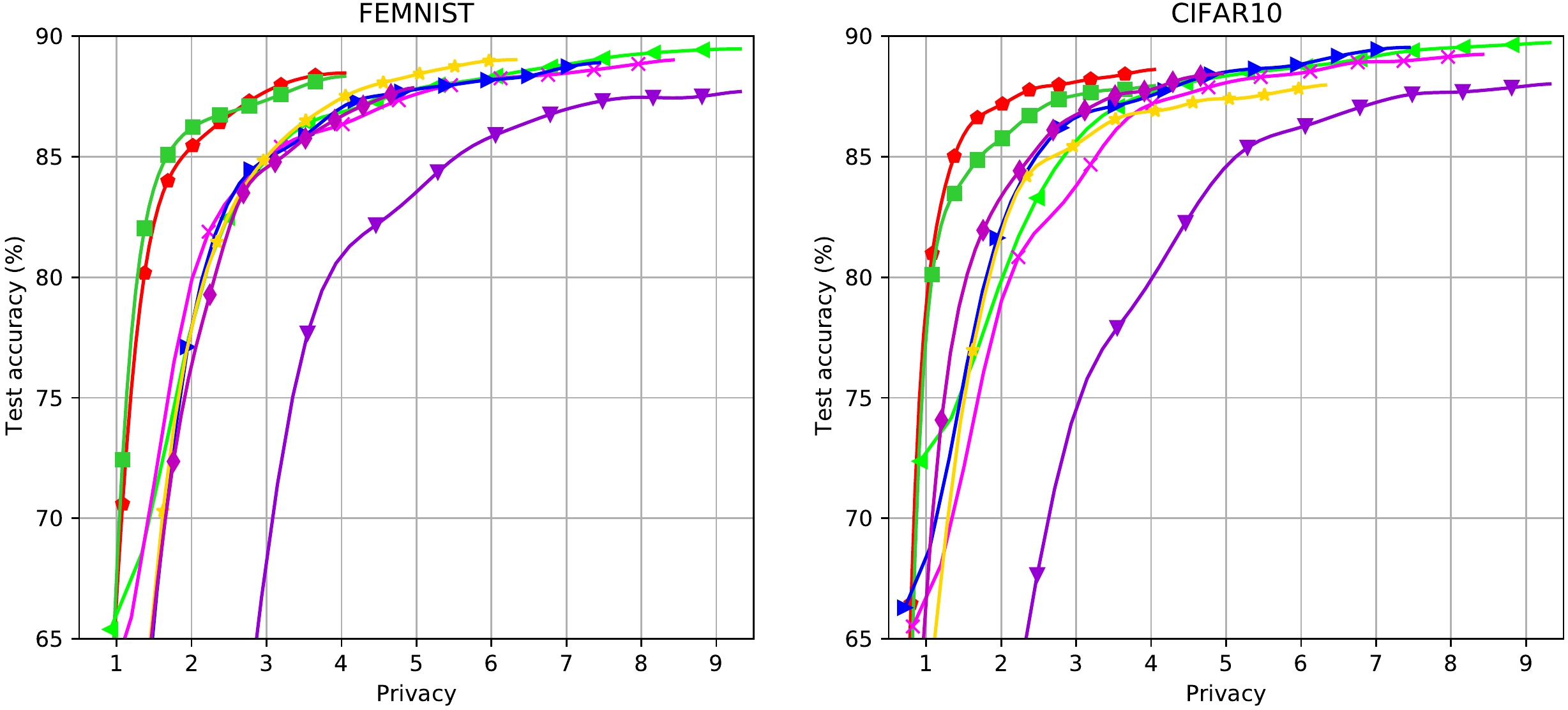}
        \label{mnist_utility_privacy_iid}
    }
    \subfigure[Non-IID setting]{
    	\includegraphics[width=0.48\textwidth]{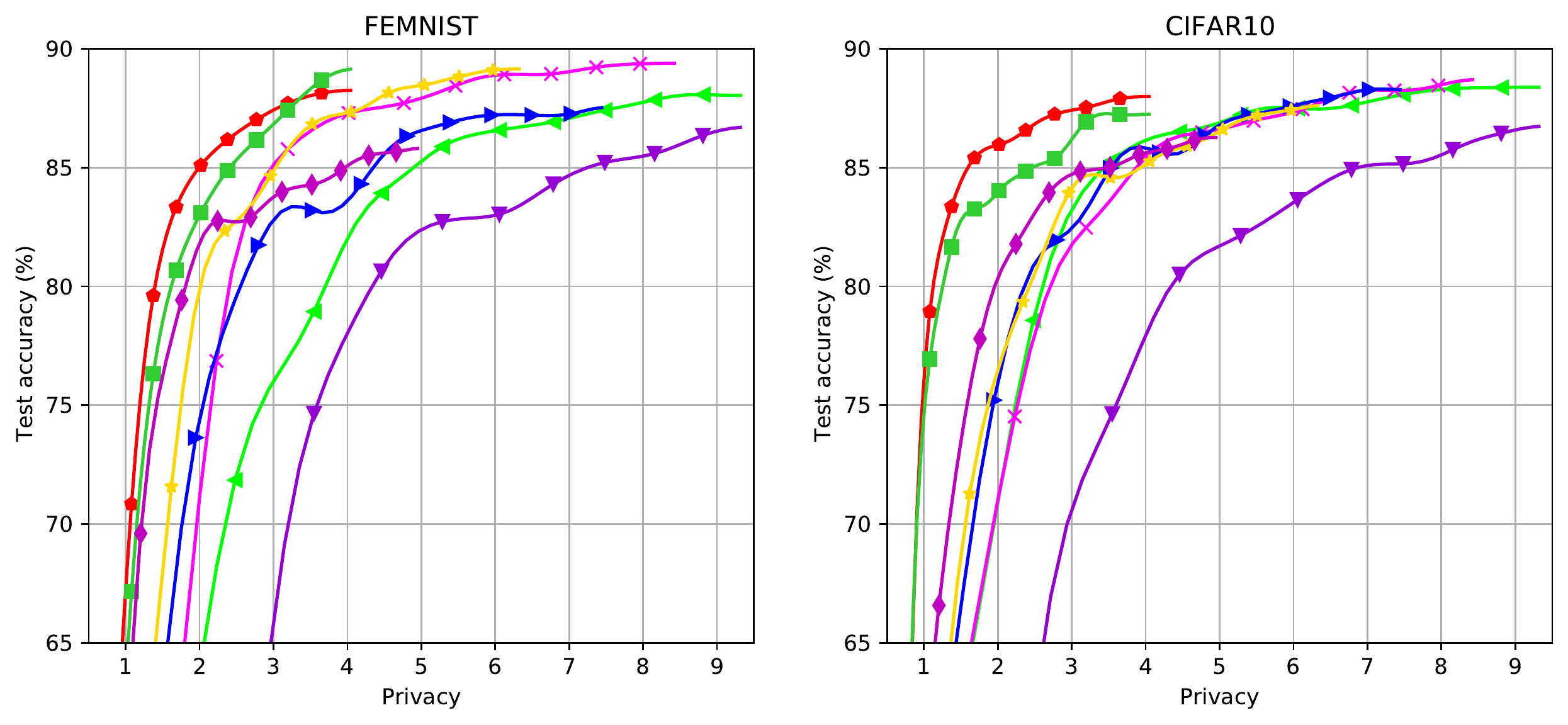}
        \label{mnist_utility_privacy_noniid}
    }
    \caption{The utility-privacy trade-off of different algorithms on FEMNIST and CIFAR10 in both IID and non-IID settings.}
    \label{mnist_utility_privacy}
\end{figure*}

\begin{figure*}[ht]
\centering
    \subfigure[IID setting]{
        \includegraphics[width=0.48\textwidth]{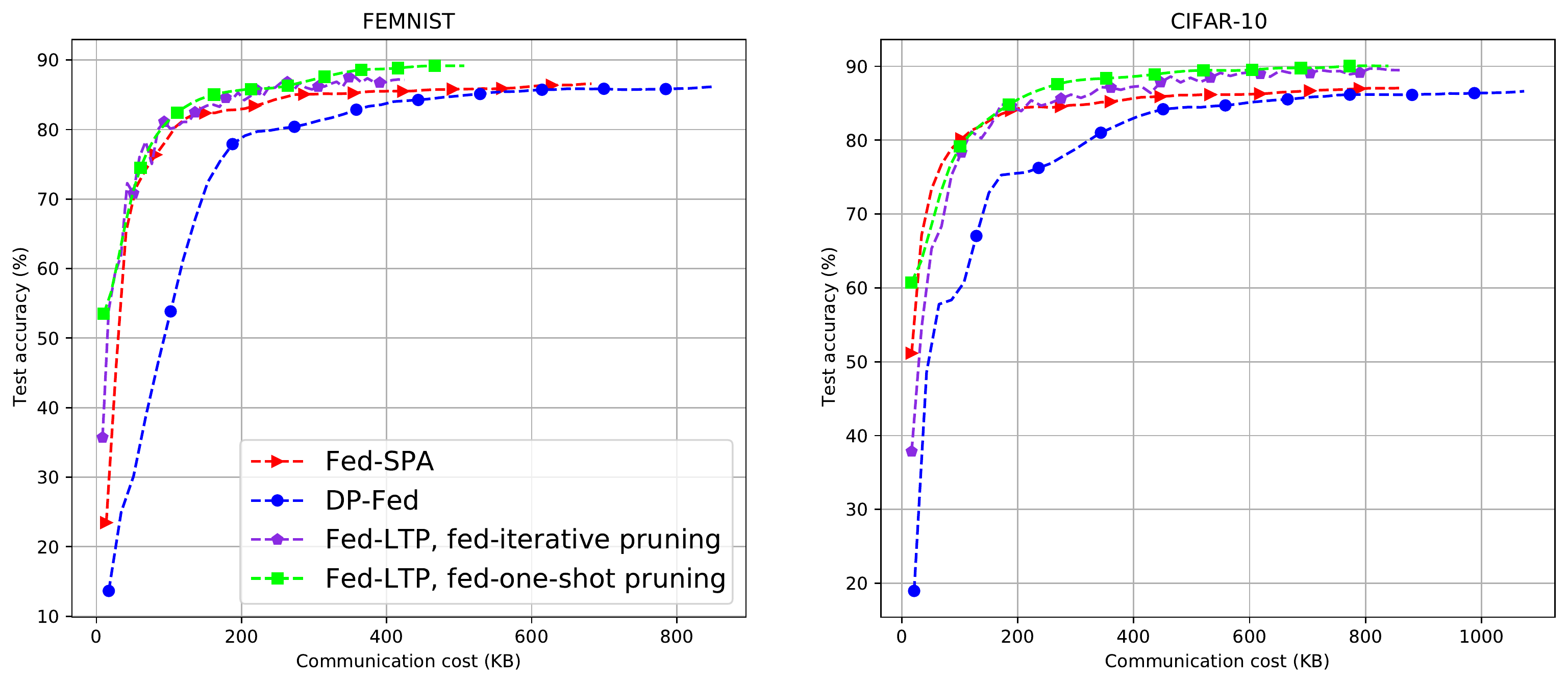}
        \label{mnist_communication_cost_iid}
    }
    \subfigure[Non-IID setting]{
    	\includegraphics[width=0.48\textwidth]{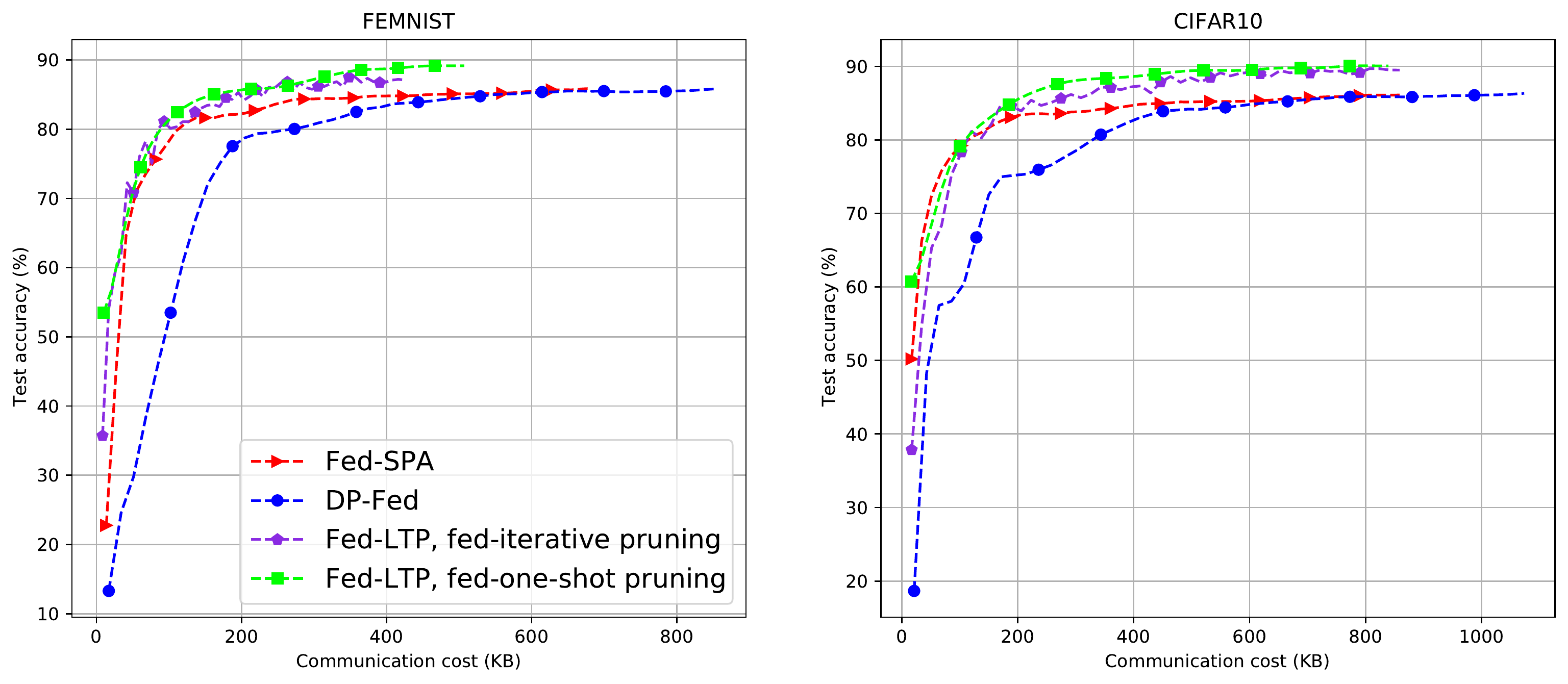}
        \label{mnist_communication_cost_noniid}
    }
    \caption{Test accuracy of different algorithms with accumulated communication cost (MB)  on two datasets in both IID and non-IID settings.}
    \label{mnist_communication_cost}
\end{figure*}

\end{document}